\documentclass[11 pt, onecolumn, draftcls]{IEEEtran}
\usepackage{epsfig}
\usepackage{amssymb,amsfonts}
\usepackage{amsmath}

\newcommand{\B}{\mathcal{B}}
\newcommand{\Real}{\mathbb{R}}

\def \bG{\mathbf{G}}
\def \bX{\mathbf{X}}
\def \bU{\mathbf{U}}

\def \ex {\mathbf{E}}
\def \bN{\mathbf{N}}
\def \bY{\mathbf{Y}}

\def \b1 {\mathbf{1}}

\begin{document}
\title{On sensing capacity of sensor networks for the class of linear observation, fixed SNR models}
\author{Shuchin~Aeron, Manqi
Zhao, and Venkatesh~Saligrama 
\thanks{}
\thanks{The authors are with the department of Electrical and Computer Engineering at
Boston University, MA -02215. They can be reached at
\{shuchin, mqzhao, srv\}@bu.edu}}


\maketitle
\newtheorem{condition}{\indent \bf Condition}[section]
\newtheorem{property}{\indent \bf Property}[section]
\newtheorem{defn}{\indent \bf Definition}[section]
\newtheorem{conj}{\indent \bf Conjecture}[section]
\newtheorem{cor}{\indent \bf Corollary}[section]
\newtheorem{lem}{\indent \bf Lemma}[section]
\newtheorem{claim}{\indent \bf Claim}[section]
\newtheorem{thm}{\indent \bf Theorem}[section]
\newtheorem{prop}{\indent \bf Proposition}[section]
\newtheorem{remark}{\indent \bf Remark}[section]
\newtheorem{example}{\indent \bf Example}[section]
\begin{abstract}
In this paper we address the problem of finding the sensing capacity
of sensor networks for a class of linear observation models and a
fixed SNR regime. Sensing capacity is defined as the maximum number
of signal dimensions reliably identified per sensor observation.
In this context sparsity of the phenomena is a key feature that
determines sensing capacity. Precluding the SNR of the environment
the effect of sparsity on the number of measurements required for
accurate reconstruction of a sparse phenomena has been widely dealt
with under compressed sensing. Nevertheless the development there
was motivated from an algorithmic perspective. In this paper our aim
is to derive these bounds in an information theoretic set-up and
thus provide algorithm independent conditions for reliable
reconstruction of sparse signals. 
In this direction we first generalize the Fano's inequality and
provide lower bounds to the probability of error in reconstruction
subject to an arbitrary distortion criteria. Using these lower
bounds to the probability of error, we derive upper bounds to
sensing capacity and show that for fixed SNR regime sensing capacity
goes down to zero as sparsity goes down to zero. This means that
disproportionately more sensors are required to monitor very sparse
events. We derive lower bounds to sensing capacity (achievable) via
deriving upper bounds to the probability of error via adaptation to
a max-likelihood detection set-up under a given distortion criteria.
These lower bounds to sensing capacity exhibit similar behavior
though there is an SNR gap in the upper and lower bounds.
Subsequently, we show the effect of correlation in sensing across
sensors and across sensing modalities on sensing capacity for
various degrees and models of correlation. Our next main
contribution is that we show the effect of sensing diversity on
sensing capacity, an effect that has not been considered before.
Sensing diversity is related to the effective \emph{coverage} of a
sensor with respect to the field. In this direction we show the
following results (a) Sensing capacity goes down as sensing
diversity per sensor goes down; (b) Random sampling (coverage) of
the field by sensors is better than contiguous location sampling
(coverage). In essence the bounds and the results presented in this
paper serve as guidelines for designing efficient sensor network
architectures.

\end{abstract}

\section{Introduction}
\label{sec:introduction} In this paper we study fundamental limits
to the performance of sensor networks for  a class of linear sensing
models under a fixed SNR regime. Fixed SNR is an important and
necessary ingredient for sensor network applications where the
observations are inevitably corrupted by external noise and clutter.
In addition we are motivated by sensor network applications where
the underlying phenomena exhibits sparsity. Sparsity is manifested
in many applications for which sensor networks are deployed, e.g.
localization of few targets in a large region, search for targets
from among a large number of sites e.g. land mine detection,
estimation of temperature variation for which few spline
coefficients may suffice to represent the field , i.e. phenomena is
sparse under a suitable transformation. More recent applications
such as that considered in \cite{heroSSP05} also involve imaging a
sparse scattering medium.

The motivation for considering linear sensing models comes from the
fact that in most cases the observation at a sensor is a
superposition of signals that emanate from different sources,
locations etc. For e.g., in seismic and underground borehole sonic
applications, each sensor receives signals that is a superposition
of signals arriving from various point/extended sources located at
different places. In radar applications~\cite{heroSSP05,mimoradar},
under a far field assumption the observation system is linear and
can be expressed as a matrix of steering vectors. 
In this case the directions becomes the variable space and one looks
for strategies to optimally search using many such radars.
Statistical modulation of gain factors in different directions is
feasible in these scenarios and is usually done to control the
statistics of backscattered data. In other scenarios the scattering
medium itself induces random gain factors in different directions.

%

In relation to signal sparsity compressive sampling,
\cite{Donoho1,Candes1} has shown to be very promising in terms of
acquiring minimal information, which is expressed as minimal number
of random projections,  that suffices for adequate reconstruction of
sparse signals. Thus in this case too, the observation model is
linear. In \cite{Nowak2} this set-up was used in a sensor network
application for realizing efficient sensing and information
distribution system by combining with ideas from linear network
coding. Also it was used in \cite{Nowak3} to build a wireless sensor
network architecture using a distributed source-channel matched
communication scheme.

For applications related to wireless sensor networks where power
limited sensors are deployed, it becomes necessary to compress the
data at each sensor. For e.g. consider a parking surveillance system
where a network of wireless low resolution cameras are deployed,
\cite{Konrad05}. With each camera taking several snapshots in space
and transmitting all of them to a base station will overwhelm the
wireless link to the base station. Instead transmission overhead is
significantly reduced by sending a weighted sum of the observations.
An illustration is shown in figure \ref{fig:ipark}. A similar set-up
was also considered in \cite{McEliece} for a robotic exploration
scenario.

\begin{figure}[t]
\centering \makebox[0in]{
    \begin{tabular}{cc}
      \psfig{figure= 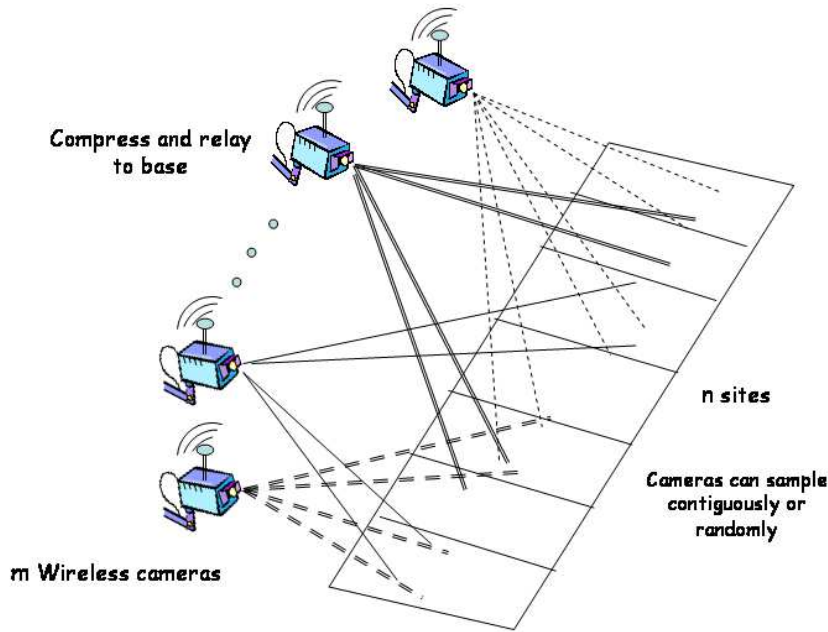,height= 3.5 in}
    \end{tabular}
  }
  \caption{A schematic of I-Park: a parking lot monitoring system.}
  \label{fig:ipark}
\end{figure}

Motivated by the scenarios considered above we start with sensing
(observation) models where at a sensor the information about the
signal is acquired as a  projection of the signal onto a weight
vector. Under this class of observation model, the sensing model is
linear and is essentially a matrix, $\bG \in \Real^{m \times n}$
chosen from some appropriate class particular to the application. In
this work we consider a fixed $SNR$ model (see also \cite{Nowak1})
where the observations at $m$ sensors for the signal $\bX \in {\cal
X}^n$ are given by,

\begin{equation} \label{eq:fixsnr}\bY = \sqrt{SNR}\,\, \bG \bX +
\bN \end{equation}

where each row of the matrix $\bG$ is restricted to have a unit
$\ell_2$ norm and where $\bN$ is the noise vector with unit noise
power in each dimension. It is important to consider fixed SNR
scenario particularly for applications related to sensor networks.
Practically each sensor is power limited. In an active sensing
scenario the sensors distribute this power to sense different
modalities, or to look (beamform) in various directions. Thus we
restrict the $\ell_2$ norm of each row of $\bG$ to be unity and then
scale the system model appropriately by $SNR$. For a networked
setting we assume that the observations made at the sensors are
available for processing at a centralized location or node. In case
when this is infeasible or costly, information can be exchanged or
aggregated at each sensor using distributed consensus type
algorithms, such as that studied in \cite{srv_murat1}.

In order utilize the information theoretic ideas and tools, we adopt
a Bayesian perspective and assume a prior distribution on $\bX$.
Another motivation for considering a Bayesian set-up is that one can
potentially model classification/detection scenarios where prior
information is usually available and is useful. Note that under some
technical conditions it can be shown that a lower bound to the
Bayesian error is also lower bound to worst case probability of
error for the parametric set-up. Therefore the lower bounds
presented in this paper also provide lower bounds to the parameter
estimation problem.

In this paper we capture the system performance via evaluating
asymptotic upper and lower bounds to the ratio $C(d_0) =
\frac{n}{m}$ such that reconstruction to within a distortion level
$d_0$ is feasible. We call the ratio $C(d_0)$ as \emph{sensing
capacity} : the number of signal dimensions reliably identified per
projection (sensor). This term was coined in \cite{Rachlin1} in the
context of sensor networks for discrete applications. Alternatively,
bounds to $C(d_0)$ can be interpreted as providing \emph{scaling
laws} for the minimal number of sensors/projections required for
reliable monitoring/signal reconstruction.

For a signal sparsity level of $k$, a different ratio of
$\frac{k}{m}$ also seems to be a reasonable choice, but in most
cases $k$ is unknown and needs to be determined, e.g., target
density, or sparsest signal reconstruction. Here it is important to
\emph{penalize false alarms, misclassification costs}. Furthermore,
$n$ and $m$ are known and part of the problem specification, while
signal complexity is governed by $k$, and one of our goals is to
understand performance as a function of signal complexity. In this
paper we show that sensing capacity $C(d_0)$ is also a function of
signal sparsity apart from $SNR$.

The upper bounds to $C(d_0)$ are derived via finding lower bounds to
the probability of error in reconstruction subject to a distortion
criteria, that apply to any algorithm used for reconstruction. The
achievable (lower) bounds to $C(d_0)$  are derived via upper
bounding the probability of error in a max-likelihood detection
set-up over the set of rate distortion quantization points. Since
most of the development for these classes of problems has been
algorithmic, \cite{Donoho1,Nowak1}, our motivation for the above
development is driven by the need to find fundamental
\emph{algorithm independent bounds} for these classes of problems.
In particular, under an i.i.d model on the components of $\bX$ that
models a priori information, e.g. sparsity of $\bX$, and letting
$\hat{\bX}(\bY)$ denote the reconstruction of $\bX$ from $\bY$, then
we show that,

\begin{eqnarray}
\label{eq:lowerboundpe} Pr\left( \frac{1}{n}d(\hat{\bX}(\bY),\bX)
\geq d_0\right)  \geq \dfrac{R_{X}(d_0) - K(d_0,n) - \frac{1}{n}
I(\bX;\bY|\bG)}{R_{X}(d_0) } - o(1)
\end{eqnarray}

for some appropriate distortion measure $d(.,.)$ and where
$R_{X}(d_0)$ is the corresponding scalar rate distortion function;
$K(n,d_0)$ is bounded by a constant and it depends on the number of
neighbors of a quantization point in an optimal $n-$dimensional rate
distortion mapping.

Next, we consider the effect of structure of $\bG$ on the
performance. Using the result on the lower bound on the probability
of error given by equation~(\ref{eq:lowerboundpe}), a necessary
condition is immediately identified in order that the reconstruction
to within an average distortion level $d_0$ is feasible, which is,
$R_{X}(d_0) - K(n,d_0) \leq \dfrac{1}{n}I(\bX;\bY|\bG)$. For a fixed
prior on $\bX$ the performance is then determined by the mutual
information term that in turn depends on $\bG$. This motivates us to
consider the effect of the structure of $\bG$ on the performance and
via evaluation of $I(\bX;\bY|\bG)$ for various ensembles of $\bG$ we
quantify the performance of many different scenarios that restrict
the choice of $\bG$ for sensing. Under the case when $\bG$ is chosen
independently of $\bX$ and randomly from an ensemble of matrices (to
be specified later in the problem set-up), we have

\begin{eqnarray}
\label{eq:structure_init}
I(\bX;\bY,\bG) & = &\underset{=0}{\underbrace{I(\bX;\bG)}} + I(\bX;\bY|\bG)\\
                & = & I(\bX;\bY) + I(\bX;\bG|\bY)\\
\label{eq:structure_end} \Rightarrow I(\bX;\bY|\bG)  & = &
I(\bX;\bY) + I(\bX;\bG|\bY)
\end{eqnarray}
This way of expanding allow us to \emph{isolate} the effect of
structure of the sensing matrix $\bG$ on the performance which in
principle influences bounds on $C(d_0)$ through the change in mutual
information as captured via the equations
\ref{eq:structure_init}-\ref{eq:structure_end} and as applied to
satisfy the necessary conditions prescribed by the lower bound in
equation~(\ref{eq:lowerboundpe}).

Using the above idea, in this paper we will show the effect of
sensing diversity on the performance, a concept which is explained
next. Under the sensing model as prescribed above, at each sensor
one can relate each component of the corresponding projection vector
as contributing towards \emph{diversity} in sensing. The total
number of non-zero components in the projection vector is called
sensing diversity. This terminology is analogous to that used in
MIMO systems in the context of communications. As will be shown
later on that loss in sensing capacity is not very significant at
reasonable levels of sensing diversity (with randomization in
sampling per sensor). In fact there is a saturation effect that
comes into play, which implies that most of the gains can be
obtained at diversity factor close to $0.5$. Now if one considers
the noiseless case, i.e. $\bY = \bG \bX$, then it was shown in
\cite{Donoho1} that for some $m$ and for some sparsity $k$ as a
function of $n$ and the coherence of the sensing matrix, an $\ell_1$
optimization problem :

\[ \begin{array}{l} \min ||\bX||_1  \\
                    \mbox{subject to}: \,\, \bY = \bG \bX, \,\, \bX
                    \geq 0
                    \end{array} \]

yields exact solution. To this end note that if $\bG$ is sparse then
solving the above system is computationally faster as is shown in
\cite{Boyd}.

There are other types of modalities that arise in the context of
resource constrained sensor networks. As an example consider the
application in \cite{Konrad05} where each camera may be physically
restricted to sample \emph{contiguous locations} in space or under
limited memory it is restricted to sample few locations, possibly at
\emph{random}. This motivates us to consider other structures on
$\bG$ under such modalities of operation. In this paper we will
contrast random sampling and contiguous sampling and show that
random sampling is better than contiguous sampling. In such
scenarios it becomes important to address a \emph{coverage} question
and in some cases may lead to a poor performance. In highly resource
constrained scenarios randomization in elements of $\bG$ is not
feasible. In this direction we also consider an ensemble of
$\left\{0,1\right\}$ matrices, with and without randomization in the
locations of non-zero entries in each row. To facilitate the reading
of the paper we itemize the organization as follows.

\begin{itemize}

\item[1.] We present the problem set-up in section
\ref{sec:probsetup} where we make precise the signal models and the
ensembles of sensing matrices that will be considered in relation to
different sensor networking scenarios.

\item[2.] In section \ref{sec:lbPe} we will present the lower bounds to
the probability of error in reconstruction subject to an average
distortion criteria. The development is fairly general and is
self-contained.

\item[3.] In section \ref{sec:ubPe} we will present a \emph{constructive} upper bound to the probability
of error in reconstruction subject to an average $\ell_2$ distortion
criteria. The development there is particular to the fixed SNR
linear sensing model that is the subject of the present paper,
though the ideas are in general applicable to other sensing models
and to other classes of distortion measures.

\item[4.] Once we establish the upper and lower bounds, we will use the
results to obtain upper and lower bounds to sensing capacity for the
fixed SNR linear sensing models, in sections
\ref{sec:upperbound_scap} and \ref{sec:lowerbound_scap}. In these
sections we will consider the full diversity Gaussian ensemble for
sensing matrix. The motivation to consider this model is that the
mutual information and moment generating functions are easier to
evaluate for the Gaussian ensemble. This is thus useful to gain
initial insights into the tradeoffs of \emph{signal sparsity} and
SNR.

\item[5.] Since the bounds to sensing capacity can be interpreted as
providing bounds for number of projections/sensors for reliable
monitoring, in section \ref{sec:compare} we will compare the scaling
implied by bounds to sensing capacity to that obtained in
\cite{Nowak1} in the context of complexity penalized regularization
framework.

\item[6.] In section \ref{sec:structure} we consider the effect of the structure
of the sensing matrix $\bG$ on sensing capacity. The section is
divided into several subsections. We begin by considering the effect
of \emph{sensing diversity} on sensing capacity. Following that we
consider the effect of correlation in the columns of $\bG$ on
achievable sensing capacity. Then we consider a very general case of
a \emph{deterministic} sensing matrix and via upper bounding the
mutual information we comment on the performance of various types of
sensing architectures of interest.

\item[7.] In section \ref{sec:ub_01ensemble} we consider the
$\left\{0,1\right\}$ ensemble for sensing matrices and provide upper
bounds to sensing capacity for various modalities in sensing.

\item[8.] In section \ref{sec:func_est} we give an example of how our
methods can be extended to handle cases when one is interested in
reconstruction of functions of $\bX$ rather than $\bX$ itself. In
this direction we will consider the case of recovery of sign
patterns of $\bX$.

\end{itemize}

\section{Problem Set-up}
\label{sec:probsetup} Assume that the underlying signal $\bX$ lies
in an n-dimensional space ${\cal X}^n$, where ${\cal X}$ can be
discrete or continuous. Discrete ${\cal X}$ models scenarios of
detection or classification and continuous ${\cal X}$ models
scenarios of estimation.

\paragraph{Fixed SNR model}: The observation model for the sensors is a linear observation
model and is given by,

\begin{eqnarray}
\label{eqn:fixedsnr} \bY = \sqrt{SNR}\,\, \bG \bX + \bN
\end{eqnarray}

which is the fixed $SNR$ model as described in the introduction. The
matrix $\bG \in \Real^{m\times n}$ is a random matrix selected from
an ensemble which we will state subsequently. For all $m,n$ each row
of $\bG$ is restricted to have a unit $\ell_2$ norm. The noise
vector $\bN$ is i.i.d. Gaussian unit variance in each dimension.

\subsection{Discussion about fixed SNR model}
At this point it is important to bring out an important distinction
of the assumption and subsequently analysis of a fixed SNR model in
contrast to similar scenarios considered but in albeit high SNR
setting. The observation model of equation \ref{eq:fixsnr} studied
in this paper is related to a class of problems that have been
central in statistics. In particular it is related to the problem of
regression for model order selection. In this context the subsets of
columns of the sensing matrix $\bG$ form a model for signal
representation which needs to be estimated from the given set of
observations. The nature selects this subset in a
weighted/non-weighted way as modeled by $\bX$. The task is then to
estimate this model order and thus $\bX$. In other words estimate of
$\bX$ in most cases is also linked to the estimate of the model
order under some mild assumptions on $\bG$. Several representative
papers in this direction are~\cite{Tibshirani96,knight00,Fan01} that
consider the performance of several (signal) complexity penalized
estimators in both parametric and non-parametric framework. One of
the key differences to note here is that the analysis of these
algorithms is done for the case when $SNR \rightarrow \infty$, i.e.
in the limit of high SNR which is reflected by taking the additive
noise variance to go to zero or not considering the noise at all.
However \emph{SNR is an important and necessary ingredient} for
applications related to sensor networks and therefore we will not
pursue a high SNR development here. Nevertheless the results
obtained are directly applicable to such scenarios.

In the next section we will first outline prior distribution(s) on
$\bX$, that reflect the sparsity of the signal $\bX$ and the model
for realizing sensing diversity in the sensing matrix $\bG$. Then we
will outline the choices of ensembles for the sensing matrix $\bG$.
In the following ${\cal N}(m,\sigma^2)$ denotes the Gaussian
distribution with mean $m$ and variance $\sigma^2$.

\subsection{Generative models of signal sparsity and sensing
diversity}
\paragraph{Signal sparsity} In a Bayesian set-up we model the
sparsity of the phenomena by assuming a mixture distribution on the
signals $\bX$. In particular the $n$ dimensional vector $\bX =
X_1,...,X_n$ is a sequence drawn i.i.d from a mixture distribution

\[ P_{X} = \alpha {\cal N}(m_1,\sigma_{1}^2) + (1 - \alpha) {\cal N}(m_0,\sigma_{0}^{2}) \]

where $\alpha \leq \frac{1}{2}$. In this paper we consider two
cases.
\begin{enumerate}
\item \textbf{Discrete Case}: $m_1 = 1$ and $m_0 = 0$ and $\sigma_1 = \sigma_0 =
0$. This means that $\bX$ is a Bernoulli$(\alpha)$ sequence. This
models the discrete case for addressing problems of target
localization, search, etc.
\item \textbf{Continuous Case}: $m_1 = m_2 = 0$ but $\sigma_{1}^{2} = 1$ and
$\sigma_{0}^{2} = 0$. This models the continuous case.
\end{enumerate}

In this context we call $\alpha$ the sparsity ratio which is held
fixed for all values of $n$. Under the above model, on an average
the signal will be $k$ sparse where $k = \alpha n$. Note that $k
\rightarrow \infty$ as $n \rightarrow \infty$.

\paragraph{Sensing diversity and ensemble for $\bG$}  In
connection to the model for diversity, the sensing matrix $\bG$ is
random matrix such that for each row $i$, $\bG_{ij}, j = 1,2,..,n$
are distributed i.i.d according to a mixture distribution, $(1 -
\beta) {\cal N}(m_0,\sigma_{0}^{2}) + \beta {\cal
N}(m_1,\sigma_{1}^{2})$. We consider three cases:
\begin{enumerate}
\item \textbf{Gaussian ensemble}: $m_1 = m_0 = 0$ and $\sigma_1 = 1 ; \sigma_0
= 0$
\item \textbf{Deterministic $\bG$}: The matrix $\bG$ is
deterministic.
\item \textbf{$\left\{0,1\right\}^{m \times n}$ ensemble}: $m_1 = 1;
m_0 = 0$ and $\sigma_1 = \sigma_0 = 0$.
\end{enumerate}

The matrix is then normalized so that each row has a unit $\ell_2$
norm. In this context we call $\beta$ as the (sensing) diversity
ratio. Under the above model, on an average each sensor will have a
diversity of $l = \beta n$. Note that $l \rightarrow \infty$ as $n
\rightarrow \infty$. Given the set-up as described above the problem
is to find upper and lower bounds to
\[ C(d_0) = \limsup \left\{\frac{n}{m} : Pr\left(\frac{1}{n} d(\hat{\bX}(\bY), \bX) > d_0\right)
\rightarrow 0 \right\}\]

where $\hat{\bX}(\bY)$ is the reconstruction of $\bX$ from
observation $\bY$ and where $d(\bX,\hat{\bX}(\bY) = \sum_{i=1}^{n}
d(X_i,\hat{X}_i(\bY))$ for some distortion measure $d(.,.)$ defined
on ${\cal X}\times {\cal X}$. In this paper we will consider Hamming
distortion measure for discrete $\bX$ and squared distortion measure
for the continuous $\bX$. Under this set-up we exhibit the following
main results:

\begin{enumerate}
\item Sensing capacity $C(d_0)$ is also a function of $SNR$, signal sparsity
and sensing diversity.

\item For a fixed SNR sensing capacity goes to zero as sparsity goes to zero.

\item Low diversity implies low sensing capacity.

\item Correlations across the columns and across the rows of $\bG$
leads to decrease in sensing capacity.

\item For the $\left\{0,1\right\}$ ensemble for sensing matrices,
sensing capacity for random sampling is higher than for contiguous
sampling.

\end{enumerate}

In the next section we will provide asymptotic lower bounds on the
probability of error in reconstruction subject to a distortion
criteria. Following that we will provide a constructive upper bound
to the probability of error. We will then use these results to
evaluate upper and lower bounds to sensing capacity. In the
following we will use $\bX$ and $X^n$ interchangeably.

\section{Bounds to the performance of estimation algorithms: lower bounds}
\label{sec:lbPe}

\begin{lem}
\label{lem:lowerbound} Given observation(s) $\bY$ for the sequence
$X^n \triangleq \left\{X_1,...,X_n\right\}$ of random variables
drawn i.i.d. according to $P_{X}$. Let $\hat{X}^n(\bY)$ be the
reconstruction of $X^n$ from $\bY$. Also is given a distortion
measure $d(X^n,\hat{X}^n(\bY)) = \sum_{i=1}^{n}
d(X_i,\hat{X}_i(\bY))$ then,

\[\begin{array}{l} Pr\left( \frac{1}{n}d(\hat{X}^n(\bY),X^n)  \geq
d_0\right)   \geq \dfrac{R_{X}(d_0) - K(d_0,n) -
\frac{1}{n}I(X^n;\bY)}{R_{X}(d_0) } - o(1)
\end{array} \]

where $K(d_0,n)$ is bounded by a constant and where $ R_{X}(d_0)$ is
the corresponding (scalar) rate distortion function for $X$.
\end{lem}

\begin{proof} See Appendix. \end{proof}
Essentially, $K(n,d_0) = \frac{1}{n} \times \log (\sharp$ neighbors
of a quantization point in an optimal n-dimensional rate-distortion
mapping). NOTE: The assumption of a scalar valued process in lemma
\ref{lem:lowerbound} is taken for the sake of simplicity. The
results are easily generalizable and can be extended to the case of
vector valued processes.

For the simpler case of discrete parameter space, the lower bound to
the minimax error in a parameter estimation framework is related to
the Bayesian error as follows,

\begin{eqnarray}
\min_{\hat{\bX}(\bY)} \max_{\bX \in \Theta}
Pr\left(\frac{1}{n}d(\bX,\hat{\bX}(Y))\geq d_0 \right) & = &
\min_{\hat{\bX}(\bY)} \max_{P_{\Theta} \in {\cal P}_{\theta}} \sum
_{\bX \in \Theta} P(\bX)
Pr\left(\frac{1}{n}d(\bX,\hat{\bX}(Y))\geq d_0\right) \\
& \geq &  \min_{\hat{\bX}(\bY)}  \sum_{\bX\in \Theta}
\pi(\bX)Pr\left(\frac{1}{n}d(\bX,\hat{\bX}(Y))\geq d_0\right)
\end{eqnarray}

where $\Theta$ is the parameter space and ${\cal P}_{\Theta}$ is the
class of probability measures over $\Theta$ and $\pi \in {\cal P}$
is any particular distribution. The above result holds true for the
case of continuous parameter space under some mild technical
conditions. Thus a lower bound to the probability of error as
derived in this paper also puts a lower bound on the probability of
error for the parametric set-up. In our set-up we will choose $\pi$
as a probability distribution that appropriately models the a priori
information on $\bX$, e.g. signal sparsity. For modeling simple
priors such as sparsity on $\bX$ one can choose distributions that
asymptotically put most of the mass uniformly over the relevant
subset of $\Theta$ and is a key ingredient in realization of the
lower bound on probability of error derived in this paper.

We have the following corollary that follows from lemma
\ref{lem:lowerbound}.
\begin{cor}
Let $X^n = X_1,..,X_n$ be an i.i.d. sequence where each $X_i$ is
drawn according to some distribution $P_{X}(x)$ and $X^n \in {\cal
X}^n$, where $|{\cal X}|$ is finite. Given observation $\bY$ about
$X^n$ we have,

\[ Pr( X^n \neq \hat{X}^n(\bY)) \geq \frac{H(X) - \frac{1}{n} I(X^n;\bY) - 1/n}{H(X) + o(1)} - o(1) \]
\end{cor}

\subsection{Tighter bounds for discrete ${\cal X}$ under hamming
distortion}

The results in the previous section can be stated for any finite $n$
without resorting to the use of AEP for the case of discrete
alphabets, with hamming distortion as the distortion measure and for
certain values of the average distortion constraint $d_0$. We have
the following lemma.
\begin{lem}
\label{lem:discrete_LB} Given observation(s) $\bY$ for the sequence
$X^n \triangleq \left\{X_1,...,X_n\right\}$ of random variables
drawn i.i.d. according to $P_{X}$.  Then for hamming under
distortion measure $d_{H}(.,.)$, for $X_i \in {\cal X},\,\, |{\cal
X}| < \infty$ and for distortion levels, $d_0 \leq (|{\cal X}| -
1)\min_{X\in {\cal X}} P_{X}$,

\[ \begin{array}{l} Pr( \frac{1}{n}d_{H}(X^n,\hat{X}^n(\bY) \geq d_0))
 \geq \dfrac{nR_{X}(d_0) - I(X^n;\bY) - 1}{n \log(|{\cal X}|) - n \left(h(d_0) +
d_0 \log (|{\cal X}|-1)\right)} \end{array} \]
\end{lem}

\begin{proof} See Appendix. \end{proof}

\subsection{Comment on the proof technique}
The proof of lemma \ref{lem:lowerbound} closely follows the proof of
Fano's inequality \cite{csiszar}, where we start with a distortion
error event based on $\frac{1}{n}d(\hat{\bX}(\bY),\bX) \geq d_0$ and
then evaluate conditional entropy of a rate-distortion mapping
conditioned on the error event and the observation $\bY$. To bound
$K(n,d_0)$, we use results in \cite{zegerIT94} for the case of
squared distortion measure.

In relation to the lower bounds presented in this paper for the
probability of reconstruction subject to an average distortion level
one such development was considered in \cite{Yannis88} in the
context of a non-parametric regression type problem.  Let $\theta$
be an element of the metric space $(d,\Theta)$. Then given
$\left\{Y_i,\bG_{i}\right\}_{i=1}^{m}$ for some random or non-random
vectors $\bG_i \in \Real^n$ and $Y_i$ being the responses to these
vectors under $\theta$. Also is given the set of conditional pdfs
given by $p_{\theta(\bG_i)}(Y_i)$ where the notation means that that
the pdfs are parametrized by $\theta(\bG_i)$. The task is to find a
lower bound on the minimax reconstruction distortion under measure
$d$, in reconstruction of $\theta$ given $\bY$ and $\bG$. In our
case one can identify $\bX \triangleq \theta$ and $\Theta \triangleq
{\cal X}^n$ with squared metric $d$. For such a set-up lower bounds
on the asymptotic minimax expected distortion in reconstruction (not
the probability of such an event) was derived in \cite{Yannis88}
using a variation of Fano's bound (see~\cite{Ibragimov81}) under a
suitable choice of worst case quantization for the parameter space
$\Theta = \left\{\mbox{space of q-smooth functions in}\,\,
[0,1]^n\right\}$ meterized with $\ell_r, \, 1 \leq r \leq \infty$
distance.

Our derivation has a flavor of this method in terms of identifying
the right quantization, namely the rate distortion quantization for
a given level of average distortion in a Bayesian setting. Although
we evaluate the lower bounds to the probability of error and not the
expected distortion itself, the lower bound on the expected
distortion in reconstruction follows immediately. Moreover our
method works for any distortion metric $d$, though in this paper we
will restrict ourselves to cases of interest particular to sensor
networks applications.

\section{Constructive upper bound to the probability of error}
\label{sec:ubPe}

In this section we will provide a constructive upper bound to the
probability of error in reconstruction subject to an average squared
distortion level. Unlike the lower bounds in this section we will
provide upper bounds for the particular observation model of
equation~(\ref{eqn:fixedsnr}). This could potentially be generalized
but we will keep our focus on the problem at hand.

To this end, given $\epsilon > 0$ and $n$, assume that we are given
the functional mapping $f(X^n)$ (or $f(\bX)$) that corresponds to
the minimal cover at average distortion level $d_0$ as given by
lemma \ref{lem:mincover}. Upon receiving the observation $\bY$ the
aim is to map it to the index corresponding index $f(\bX)$, i.e. we
want to detect which distortion ball the true signal belongs to.
Clearly if $\bX$ is not typical there is an error. From lemma
\ref{lem:AEP}, the probability of this event can be bounded by an
arbitrary $\delta > 0 $ for a large enough n. So we will not worry
about this a-typical event in the following.
%

Since all the sequences in the typical set are equiprobable, we
covert the problem to a  max-likelihood \emph{detection} set-up over
the set of rate-distortion quantization points given by the minimal
cover as follows. Given $\bG$ we and the rate distortion points
corresponding to the functional mapping $f(X^n)$, we enumerate the
set of points, $\bG Z_{i}^{n} \in \Real^{m}$. Then given the
observation $\bY$ we map $\bY$ to the nearest point (in $\Real^m$)
$\bG Z_{i}^{n}$. Then we ask the following probability,

\[ \begin{array}{l}Pr \left(\sqrt{SNR}\bG f(\bX) \rightarrow \sqrt{SNR} \bG f(\bX') | \bG, \bX \in \B_{i}
, \bX' \in \B_{j}: \frac{1}{n} d_{set}(\B_{i},\B_{j}) \geq 2 d_0
\right)

\end{array}\]

that is, we are asking what is the probability that the in typical
max-likelihood detection set-up we will map signals from distortion
ball $\B_i$ to  signals in distortion ball $\B_j$ that is at an
average set distance $\geq 2 d_0$ from $\B_i$, where
$d_{set}(\B_i,\B_j) = \min_{\bX \in \B_i, \bX' \in \B_j}
d(\bX,\bX')$. For sake of brevity we denote the above probability
via $P_{e}(pair)$ to reflect it as a pairwise error probability.
Since the noise is additive Gaussian noise we have

\[
P_e(pair) =  Pr\left( \bN^T \bG (\bX - \bX')  \geq \frac{1}{2}
\sqrt{SNR}||\bG (\bX - \bX')||^{2} \,\,\, :  \bX \in \B_{i}, \bX'
\in \B_j\right) \]

\[
P_e(pair) =  Pr\left( \bN^T \frac{\bG (\bX - \bX')}{||\bG(\bX -
\bX')||} \geq \frac{\sqrt{SNR}}{2 ||\bG(\bX_1 - \bX_2)|| } ||\bG
(\bX - \bX')||^{2} \,\,\, : \bX \in \B_{i}, \bX' \in \B_j\right)
\]

Since noise $\bN$ is AWGN noise with unit variance in each
dimension, its projection onto the unit vector $\frac{\bG (\bX -
\bX')}{||\bG(\bX - \bX')||}$ is also Gaussian with unit variance.
Thus we have

\[
P_e(pair) =  Pr\left( N  \geq \frac{\sqrt{SNR}}{2} ||\bG (\bX -
\bX')|| \,\,\, : \bX \in \B_{i}, \bX' \in \B_j\right)
\]

By a standard approximation to the ${\cal Q}(.)$ (error) function,
we have that,

\[ P_e\left(f(\bX) \rightarrow f(\bX') | \bX \in \B_i, \bX' \in \B_j, \bG\,\,: \frac{1}{n}
d_{set}(\B_{i},\B_{j}) \geq 2 d_0 \right) \leq \exp \left\{-
\dfrac{SNR ||\bG (\bX - \bX')||^2}{4} \right\} \]

In the worst case we have the following bound,

\[ P_e\left(f(\bX) \rightarrow f(\bX') | \bX \in \B_i, \bX' \in \B_j, \bG\,\,: \frac{1}{n}
d_{set}(\B_{i},\B_{j}) \geq 2 d_0 \right) \leq \exp
\left\{- \min_{\bX \in \B_i,\bX' \in \B_j} \dfrac{SNR ||\bG (\bX -
\bX')||^2}{4} \right\} \]

Now note that from above construction it implies that the average
distortion in reconstruction of $\bX$ is bounded by $2d_0$ if the
distortion metric obeys triangle inequality. To evaluate the total
probability of error we use the union bound to get,

\[ Pr\left(\frac{1}{n}d(\bX,\hat{\bX}(\bY)) \geq 2 d_0\right) \leq \exp
\left\{- \min_{\bX \in \B_i,\bX' \in \B_j} \dfrac{ SNR||\bG (\bX -
\bX')||^2}{4} \right\} 2^{n (R_{X}(d_0) - K(n,d_0))}  \]

We will use this general form and apply it to particular cases of
ensembles of the sensing matrix $\bG$. In the following sections we
begin by providing upper and lower bounds to the sensing capacity
for the Gaussian ensemble for full diversity.

\section{Sensing Capacity: Upper bounds, Gaussian ensemble}
\label{sec:upperbound_scap}
\subsection{Discrete $\bX$, full diversity, Gaussian ensemble}
For this case we have the following main lemma.

\begin{lem}
\label{lem:scap_binom} Given $\bX \in \left\{0,1\right\}^n$ drawn
Bernoulli~$(\alpha,1-\alpha)$ and $\bG$ chosen from the Gaussian
ensemble. Then, with the distortion measure as the hamming
distortion, for a diversity ratio of $\beta = 1$ and for $d_0 \leq
\alpha$,  the sensing capacity $C$ is upper bounded by

\[ C(d_0) \leq \dfrac{\frac{1}{2} \log(1 + \alpha SNR)}{R_{X}(d_0)} \]
\end{lem}

\begin{proof}
From lemma \ref{lem:discrete_LB} the probability of error is lower
bounded by zero if the numerator in the lower bound is negative,
this implies for any $m,n$ that

\[ C_{m,n}(d_0, \bG) \leq \dfrac{\frac{1}{m} I(\bX;\bY|\bG)}{R_{X}(d_0)} \]

Since $\bG$ is random we take expectation over $\bG$.
It can be shown that the mutual information

\[\begin{array}{l} \mathbf{E}_{\bG} I(X^n;\bY|\bG)
\leq \\ \max_{P_{\bX}: \sum \frac{1}{n} \ex X_{i}^{2} \leq \alpha}
\frac{1}{2} \mathbf{E}_{\bG} \log det(\mathbf{I}_{m\times m} + \bG
\bX \bX^T \bG^T) \end{array} \] $   = \ex_{\lambda_1,..,\lambda_m}
\sum_{i =1}^{m} \frac{1}{2}\log( 1 + \lambda_i \alpha SNR)$ where
$\lambda_i$ are singular values of $\bG\bG^T$. Since rows of $\bG$
have a unit norm $\Rightarrow \lambda_i \leq 1\,\, \forall i$. Hence
$\ex_{\bG} I(X^n;\bY|\bG) \leq \frac{m}{2} \log(1 + \alpha SNR)$.
Thus the result follows.
\end{proof}

\begin{figure}[t]
\begin{centering}
\includegraphics[width = 3 in]{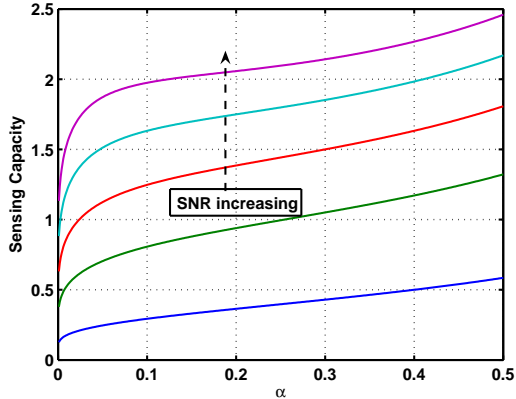}
\caption{\small{The plot of sparsity versus upper bounds to the
sensing capacity for various SNRs  for the binary case
($\mathcal{X}=\{0,1\}$) for zero Hamming distortion.}}
\label{fig:scapacity}
\end{centering}
\end{figure}

\subsection{Continuous $\bX$, full diversity, Gaussian ensemble}
\begin{lem}
\label{lem:scap_continuous} Given $\bX \in \Real^n$ drawn i.i.d.
according to $ P_{X} = \alpha {\cal N}(0,1) + (1 - \alpha) {\cal
N}(0,0) $ and $\bG$ chosen from the Gaussian ensemble. Then, for
squared distortion measure, for diversity ratio $\beta =1$ and for
$d_0 \leq \frac{\alpha}{2}$, the sensing capacity $C(d_0)$ obeys,

\[C(d_0) \leq \dfrac{\frac{1}{2}\log(1 + \alpha SNR)}{H(\alpha) + \frac{\alpha}{2}\log\frac{\alpha}{2d_0}} \]
\end{lem}

\begin{proof}
From lemma \ref{lem:scap_binom} we have that $\ex_{\bG}
I(\bX;\bY|\bG) \leq \frac{m}{2} \log(1 + \alpha SNR)$. In order that
the probability of error be lower bounded by zero, from lemma
\ref{lem:lowerbound} it follows that asymptotically

\[ \frac{n}{m} \leq \frac{\ex_{\bG} I(\bX;\bY|\bG)}{R_{X}(d_0) - K(d_0,n)}\]
It can be shown that $ |K(d_0,n) - \log 2| < \epsilon$ with
$\epsilon$ very small for large enough $n$, see e.g.
\cite{zegerIT94}. The lemma then follows by plugging in the results
from section \ref{subsec:ratedist_gauss}.
\end{proof}

It can be easily seen that as $\alpha \downarrow 0$ the sensing
capacity goes to zero. We illustrate this by plotting the upper
bounds in figure \ref{fig:scapacity} for the discrete case. We will
revisit this phenomena in section \ref{sec:compare} in relation to
the bounds derived in \cite{Nowak2}  in the context of compressed
sensing.

\section{Sensing Capacity: Lower bounds, Gaussian ensemble}
\label{sec:lowerbound_scap}
\subsection{Discrete alphabet, full diversity}
The discrete $\bX$ with hamming distortion is a special case where
we can provide tighter upper bounds. The proof follows from the
development in section \ref{sec:ubPe} and identifying that for the
discrete case one can choose the discrete set of points instead of
the distortion balls. We have the following lemma.

\begin{lem}
\label{lem:lb_scap_discrete} Given $\bX \in {\cal X}^n$ with $|{\cal
X}| < \infty$, for $\beta =1$ and $\bG$ chosen from a Gaussian
ensemble. Then for $d_0 \leq \min_{x \in{\cal X}} P_X(x)$, a sensing
capacity of
\[ C(d_0) = \dfrac{\frac{1}{2} \log (1 + \frac{SNR d_0}{2})}{H(X)- d_0 \log
|{\cal X} -1| - d_0 \log \frac{1}{d_0}} \]

is achievable in that the probability of error goes down to zero
exponentially for choices of $C = \frac{n}{m} = C(d_0) - \eta$ for
any $\eta > 0$.
\end{lem}
\begin{proof}
We have

\[ Pr\left(\frac{1}{n}d(\bX,\hat{\bX}(\bY)) \geq d_0 | \bG \right )   \leq \exp
\left\{- \dfrac{SNR ||\bG (\bX - \bX')||^2}{4} \right\} 2^{n H(X) -
n d_0 \log |{\cal X} -1| - \log \binom{n}{n d_0}}
\]

where we have applied the union bound to all the \emph{typical}
sequences that are outside the hamming distortion ball of radius
$d_0$. Taking the expectation with respect to $\bG$ we get,

\[ Pr\left(\frac{1}{n}d(\bX,\hat{\bX}(\bY)) \geq d_0 \right )   \leq \ex_{\bG} \exp
\left\{- \dfrac{SNR ||\bG (\bX - \bX')||^2}{4} \right\} 2^{n H(X) -
n d_0 \log |{\cal X} -1| - \log \binom{n}{n d_0}}
\]

Now note that since $\bG$ is a Gaussian random matrix where each row
has a unit $\ell_2$ norm, $||\bG (\bX - \bX')||^2 = \sum_{i=1}^{m}
|\sum_{j=1}^{n} \bG_{ij}(X_i - X'_j)|^2$ is a sum of $m$ independent
$\chi^2$ random variables with mean $||\bX - \bX'||^2$. Thus from
the moment generating function of the $\chi^2$ random variable we
get that,

\[ Pr\left(\frac{1}{n}d(\bX,\hat{\bX}(\bY)) \geq d_0 \right )   \leq \left(\dfrac{1}{1 + \frac{SNR ||\bX - \bX'||^2}{2\,n}}\right)^{m/2} 2^{n H(X) -
n d_0 \log |{\cal X} -1| - \log \binom{n}{n d_0}}
\]

This implies,

\[ Pr\left(\frac{1}{n}d(\bX,\hat{\bX}(\bY)) \geq d_0 \right )   \leq 2^{- \frac{m}{2} \log (1 + \frac{SNR d_0}{2})} 2^{n H(X) -
n d_0 \log |{\cal X} -1| - \log \binom{n}{n d_0}}
\]

Now note that for $d_0 \leq \alpha$, $\log \binom{n}{n d_0} \geq n
d_0 \log \frac{1}{d_0}$. Then from above one can see that the
probability of error goes down to zero if,

\[ \frac{n}{m} < \dfrac{\frac{1}{2} \log (1 + \frac{SNR d_0}{2})}{H(X)- d_0 \log
|{\cal X} -1| - d_0 \log \frac{1}{d_0}} \]

Thus a sensing capacity of

\[ C(d_0) = \dfrac{\frac{1}{2} \log (1 + \frac{SNR d_0}{2})}{H(X)- d_0 \log
|{\cal X} -1| - d_0 \log \frac{1}{d_0}} \]

is achievable in that the probability of error goes down to zero
exponentially for choices of $C = \frac{n}{m} = C(d_0) - \eta$ for
any $\eta > 0$.

\end{proof}
\subsection{Continuous $\bX$, full diversity}
\begin{lem}
\label{lem:lb_scap_continuous} [Weak Achievability] For $\bX \in
\Real^n$ and drawn i.i.d. according to $P_{x}(X)$, $\bG$ chosen from
the Gaussian ensemble and $\beta =1$, a sensing capacity of

\[ C(2 d_0)  = \dfrac{\frac{1}{2} \log(1 + d_0 SNR)}{R_{X}(d_0) -
K(n,d_0)} \]

is achievable in that the probability of error goes down to zero
exponentially with $n$ for $C =\frac{n}{m} \leq C(2 d_0) - \epsilon$
for some arbitrary $\epsilon > 0$.
\end{lem}

\begin{proof} For this case we invoke the construction as outlined in section
\ref{sec:ubPe}. From the results in that section we get that,

\[ Pr\left(\frac{1}{n}d(\bX,\hat{\bX}(\bY)) \geq  2 d_0\right) \leq \exp
\left\{- \min_{\bX \in \B_i,\bX' \in \B_j} \dfrac{SNR ||\bG (\bX -
\bX')||^2}{4} \right\} 2^{n (R_{X}(d_0) - K(n,d_0))}  \]

Note that the result is little weaker in that guarantees are only
provided to reconstruction within $d_0$, but one can appropriately
modify the rate distortion codebook to get the desired average
distortion level. Proceeding as in the case of discrete $\bX$ and ,
by taking the expectation over $\bG$ and noting that $\min_{\bX \in
\B_i,\bX' \in \B_j} ||\bX - \bX'||^2 \geq 2nd_0$, we get that,

\[ Pr\left(\frac{1}{n}d(\bX,\hat{\bX}(\bY)) \geq 2 d_0\right) \leq
\left(\dfrac{1}{1 +  SNR d_0}\right)^{m/2} 2^{n (R_{X}(d_0) -
K(n,d_0))}
\]

This implies,
\[ Pr\left(\frac{1}{n}d(\bX,\hat{\bX}(\bY)) \geq 2d_0\right) \leq
\left(\dfrac{1}{1 +  SNR d_0}\right)^{m/2} 2^{n (R_{X}(d_0) -
K(n,d_0))}
\]

\[ Pr\left(\frac{1}{n}d(\bX,\hat{\bX}(\bY)) \geq 2d_0\right) \leq
2^{- \frac{m}{2}\log(1 +  SNR d_0) } 2^{n(R_{X}(d_0) - K(n,d_0))}
\]

This implies that for

\[ \frac{n}{m} < \dfrac{\frac{1}{2} \log(1 + d_0 SNR)}{R_{X}(d_0) - K(n,d_0)}
\]
the probability of error goes to zero exponentially. This means that
a sensing capacity of

\[ C(2 d_0)  = \dfrac{\frac{1}{2} \log(1 + d_0 SNR)}{R_{X}(d_0) -
K(n,d_0)} \]

is achievable in that the probability of error goes down to zero
exponentially with $n$ for $C =\frac{n}{m} \leq C(d_0) - \eta$ for
some arbitrary $\eta > 0$.

\end{proof}

A plot of upper and lower bounds are shown in figure \ref{fig:ublb_scap}.

\begin{figure}[t]
\centering \makebox[0in]{
    \begin{tabular}{cc}
      \psfig{figure= 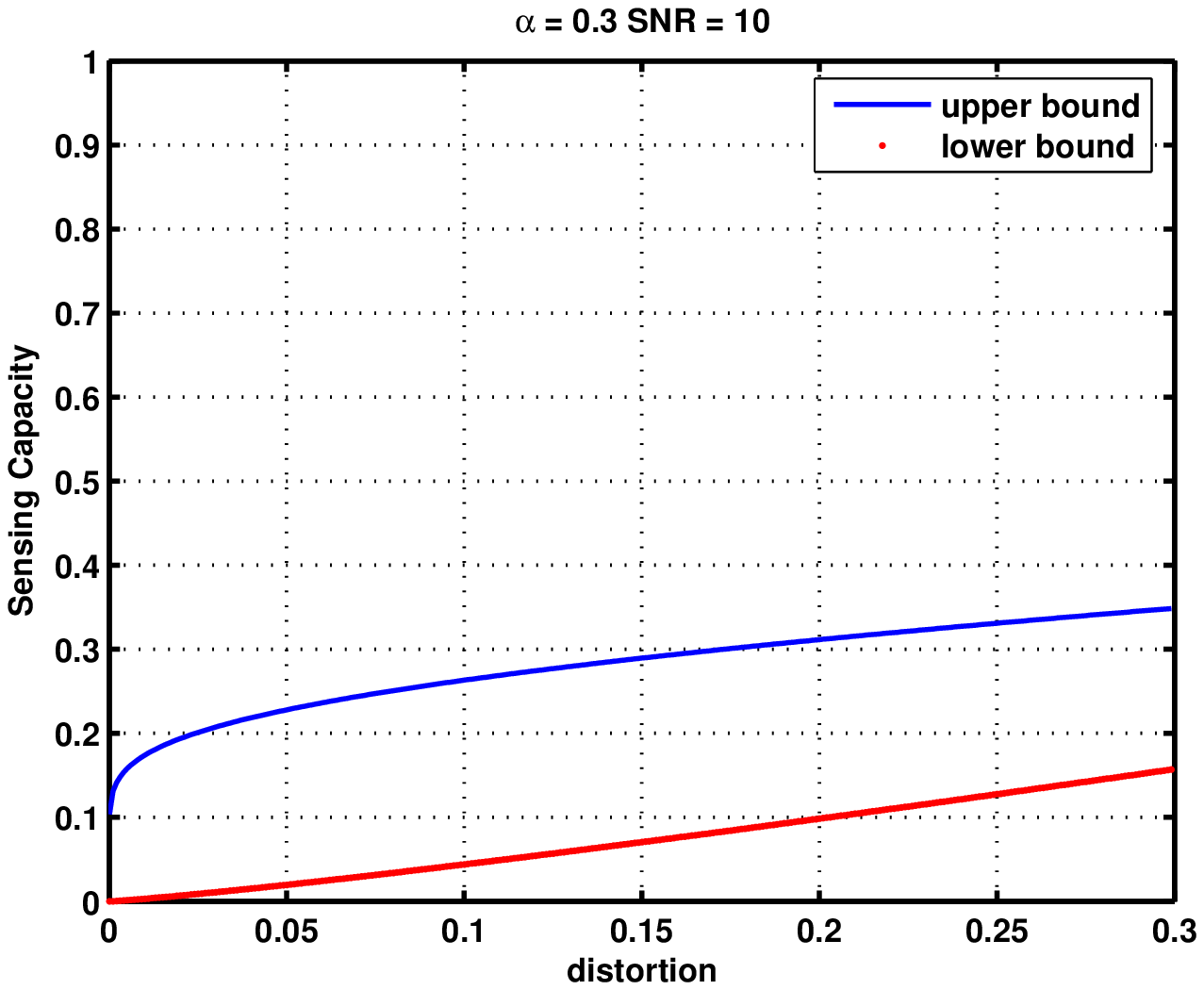,height= 2.3 in} & \psfig{figure= 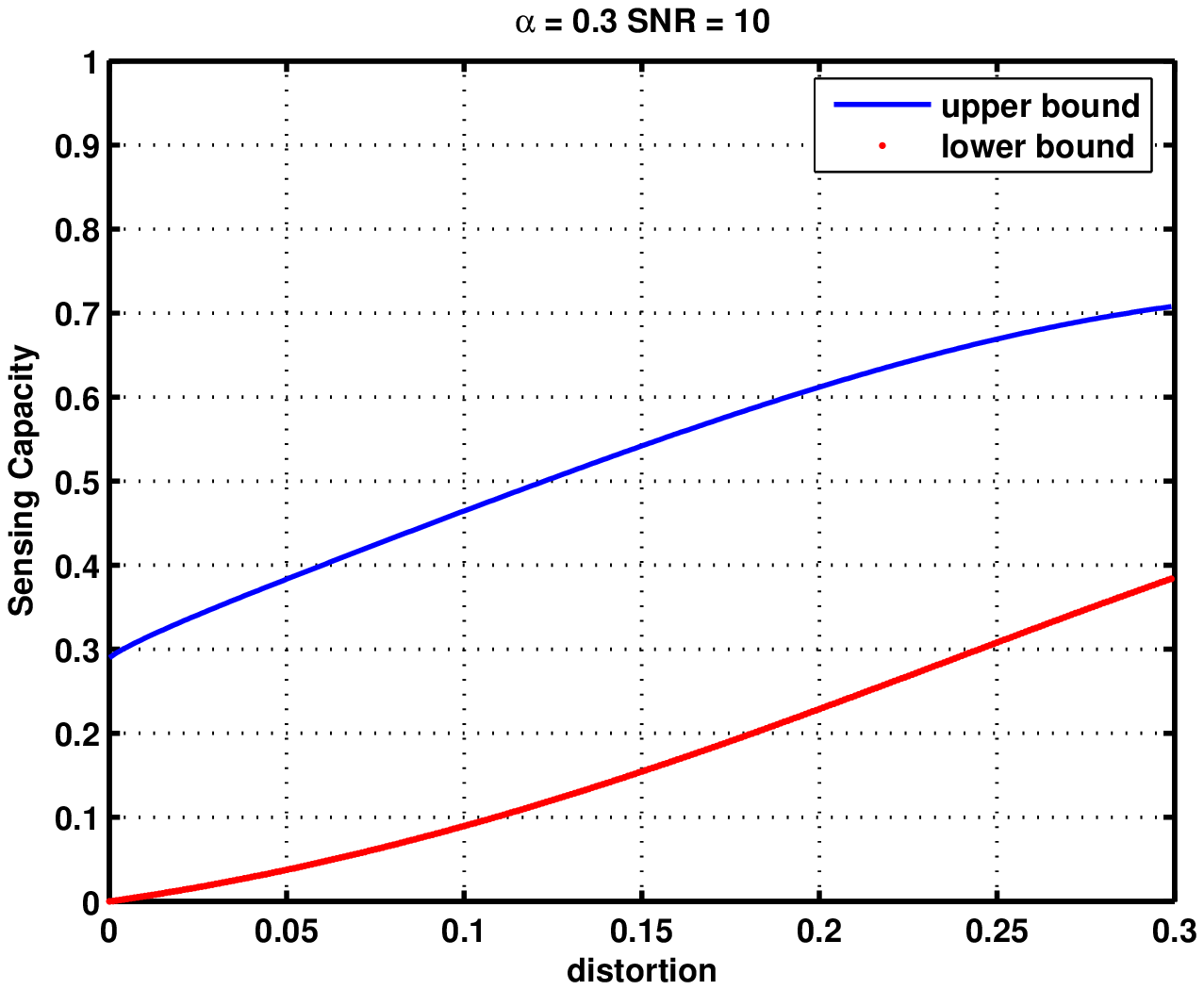,height= 2.3 in}\\
      (a) & (b)
    \end{tabular}
  }
  \caption{(a) Plots of upper and lower bounds to sensing capacity for the Gaussian mixture model. (b) Plots of upper and lower bounds for sensing capacity for the Bernoulli model. The distortion on the x-axis is mean squared distortion for the Gaussian case and hamming distortion for the Bernoulli case. Note that zero distortion achievable sensing capacity is zero and there is an SNR gap in the upper and lower bounds. }
  \label{fig:ublb_scap}
\end{figure}

\section{Comparison with existing bounds}
\label{sec:compare} Note that the results in this paper are stated
for $d_0 \leq \alpha$ for the discrete case and for $d_0 \leq
\frac{\alpha}{2}$ for the continuous case. This is because one must
consider stricter average distortion measures as the phenomena
becomes sparser. To bring out this point concretely and for purposes
of comparison with existing bounds, we consider the result obtained
in \cite{Nowak2} based on optimal complexity regularized estimation
framework. They show that the expected mean squared error in
reconstruction is upper bounded by,

\begin{eqnarray} \label{eq.nowak}
\ex \left[\frac{1}{n} ||\bX - \hat{\bX}||^2\right] \leq C_1C_2 \frac{k \log n}{m}
\end{eqnarray}

where $C_1 \sim 1$ and $C_2 \sim 50( P + \sigma)^2 \left\{(1 + p)\log 2 + 4\right\}$,
under normalization of the signal and the noise power and
$p$ is the number of quantization levels, \cite{Nowak1}. To this end consider an extremely sparse
case, i.e., $k=1$. Then the average distortion metric in equation \ref{eq.nowak},
 does not adequately capture the performance, as one can
always declare all zeros to be the estimated vector and the distortion
then is upper bounded by ${\cal O}(\frac{1}{n})$.
Consider the case when $\bX$ is extremely sparse, i.e.  $\alpha \downarrow 0$ as $\frac{1}{n}$.
Then a right comparison is to evaluate the average distortion per number of non-zero elements,
$\ex \left[\frac{1}{\alpha n} ||\bX - \hat{\bX}||^2\right]$. Using this as the performance metric we have from equation \ref{eq.nowak},

\begin{eqnarray} \label{eq.nowak1}
\ex \left[\frac{1}{\alpha n} ||\bX - \hat{\bX}||^2\right] \leq C_1C_2 \frac{n \log n}{m}
\end{eqnarray}
When $\alpha$ is small then the average number of projections
required such that the per non-zero element distortion is bounded by
a constant, scales as ${\cal O}(n \log n)$. This is indeed
consistent with our results, in that the Sensing Capacity goes down
to zero as $\frac{1}{\log n}$.\\

\noindent $\bX$ is sparse, i.e. $\alpha <1 $ but not very small.
From results on achievable sensing capacity we have that

\[\begin{array}{ll} Pr\left(\frac{1}{n} ||\bX - \hat{\bX}||^2  \geq d_0\right) \leq
-\frac{m}{2} \log(1 + d_0 SNR/2) + n (R_{X}(d_0) - K(n,d_0))
\end{array}\]

In order to compare the results we fix, performance guarantee  of
$Pr(d(\bX,\hat{\bX})  \geq d_0) \leq \epsilon $ for a given
$\epsilon >0$, we have for the minimal number of projections
required that,

\[m \geq \frac{2 \left(\log (1/\epsilon) + n (R_X(d_0)- K(n,d_0))\right)}{\log(1 + d_0 SNR/2)}
\]

from our results. From results in \cite{Nowak1} it follows that,

\[ m \geq C_1 C_2 \frac{\alpha n \log n}{d_0 \epsilon} \]

For the special case of binary alphabet we have the following
scaling orders for the number of projections in both cases, from
achievable sensing capacity we have $m_1 \geq {\cal O}(n
H_2(\alpha))$ and from results in \cite{Nowak1}  we have $m_2 \geq
{\cal O}(\alpha n \log n)$. A plot of these orders as a function of
$\alpha$ for a fixed $n$ is shown in figure, \ref{fig:nowak}.

\begin{figure}[t]
\begin{centering}
\includegraphics[width = 3.5 in]{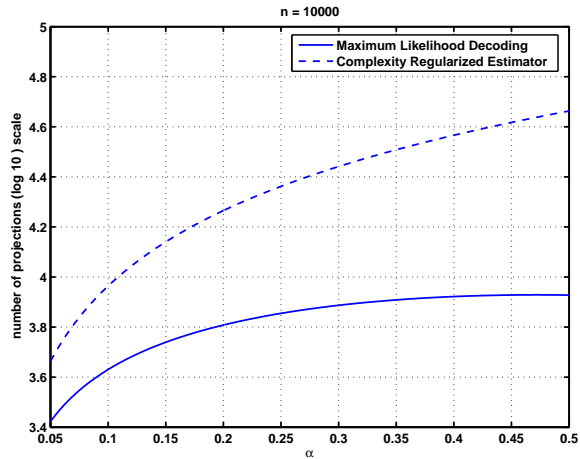}
\caption{\small{The difference in scaling of the number  of
projections with the sparsity rate from bounds derived from Sensing
Capacity and from bounds obtained in \cite{Nowak1}. Our bounds are
sharper. }} \label{fig:nowak}
\end{centering}
\end{figure}

\section{Effect of structure of $\bG$}
\label{sec:structure}

In this section we will show that effect of structure of $\bG$ on
sensing capacity. This section is divided into several subsections
and the discussion is self-contained. In section
\ref{subsec:diversity_gauss} we will show that for the Gaussian
ensemble, the sensing capacity reduces for when diversity is low.
Following that in section \ref{subsec:correlation} we will show the
effect of correlation across columns in the sensing matrix for the
Gaussian ensemble on achievable sensing capacity. In section
\ref{subsec:deterministic} we will present a general result for a
\emph{generic} sensing matrix $\bG$ which will subsequently be used
to highlight the effect of structures such as that induced via
random filtering using a FIR filter with/without downsampling as
considered in \cite{Wakin06}.

\subsection{Effect of sensing diversity, Gaussian ensemble}
\label{subsec:diversity_gauss} In order to show the effect of
sensing diversity we evaluate the mutual information $\ex_{\bG}
I(\bX;\bY|\bG)$ using the intuition described in the introduction.
To this end we have the following lemma.

\begin{lem}
\label{lem:bound_info2} For a diversity ratio of $\beta$, with $l =
\beta n$ as the average diversity per sensor and an average sparsity
level of $k = \alpha n$ , we have
\begin{eqnarray}\label{{ap2}}
\ex_{\bG}
I(\bX;\bY|\bG)\leq\frac{m}{2}\mathbf{E}_j\left[\log\left(\frac{SNR}{l}j+1\right)\right],
\end{eqnarray}
where the expectation is evaluated over the distribution
\[\Pr(j)=\dfrac{\binom{k}{j}\binom{n-k}{l-j}}{\binom{n}{l}}\]
\end{lem}

\begin{proof} See Appendix. \end{proof}

In the above lemma $j$ plays the role of number of overlaps between
the projection vector and the sparse signal. As the diversity
reduces this overlap reduces and the mutual information decreases.
We will illustrate this by considering the extreme case when $\beta
\downarrow$ with $n$ as $\frac{1}{n}$.  For this case we have,
\[\begin{array}{l}
I(\bX;\bY|\bG)\\
\leq \frac{m}{2}\mathbf{E}_j\left[\log\left(\frac{j\,\,SNR}{l}+1\right)\right]\\
=\frac{m}{2}[(1-\alpha)\log(SNR \cdot 0+1)+\alpha
\log(SNR+1)]\\
= \frac{m\alpha}{2}\log(1+ SNR) \end{array}
\]

\begin{figure}[t]
\begin{centering}
\includegraphics[width = 4 in]{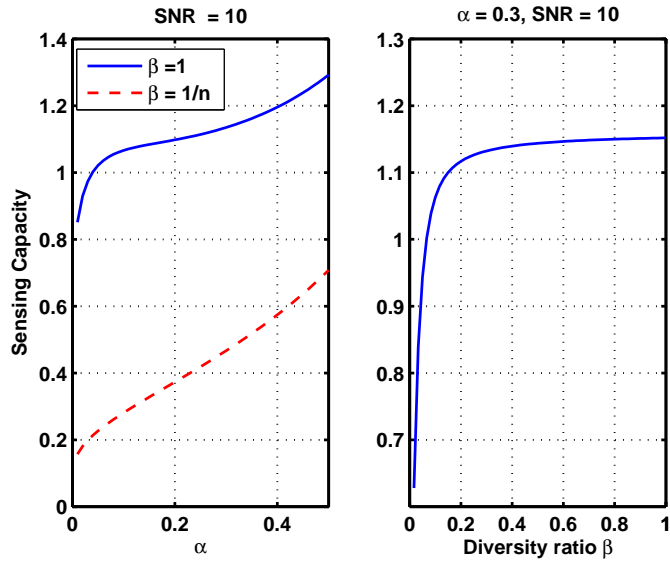}
\caption{\small{The gap between upper bounds to sensing capacity in
very low diversity and full diversity for the binary alphabet case.
Shown also is the Sensing Capacity as a function of diversity for
fixed sparsity. Note the saturation effect with diversity ratio.}}
\label{fig:diversity}
\end{centering}
\end{figure}

The effect is illustrated in figure \ref{fig:diversity}. Thus low
sensing diversity implies low sensing capacity.

\subsection{Effect of correlation in $\bG$ on achievable sensing
capacity} \label{subsec:correlation}

In this section we will show that correlation in sensing matrix
$\bG$ reduces achievable capacity. Correlation in $\bG$ can arise
due to many physical reasons such as correlated scattering,
correlation of gains across modalities in sensing which may arise
due to the physical construction of the sensor. Naturally there can
be direct relations between various phenomena that can lead to such
correlation. This is captured by assuming that there is correlation
across the columns of $\bG$. Consider the upper bound to the
probability of error as derived in section \ref{sec:ubPe},

\[ Pr\left(\frac{1}{n}d(\bX,\hat{\bX}(\bY)) \geq 2 d_0\right) \leq \exp
\left\{- \min_{\bX \in \B_i,\bX' \in \B_j} \dfrac{ SNR||\bG (\bX -
\bX')||^2}{4} \right\} 2^{n (R_{X}(d_0) - K(n,d_0))}  \]

In the above expression, the term
\[SNR||\bG (\bX - \bX')||^2 = SNR
\sum_{i=1}^{n}| \sum_{j=1}^{n} \bG_{ij} (X_i - X'_j)|^2 \]

where $\sum_{j=1}^{n} \bG_{ij} (X_i - X'_j)$ for each $i$ are
independent Gaussian random variables with zero mean and variance
given by- $\Delta^{T} \Sigma_{\bG_i} \Delta$ where $\Delta$ is the
vector $\Delta = \bX - \bX'$ and $\Sigma_{G_i}$ is the covariance
matrix (symmetric and positive semi-definite) of the $i$-th row of
$\bG$. By construction, we know that $\frac{1}{n} \Delta^{T} \Delta
\geq 2 d_0$ and note that in the worst case,

\[\min \,\, \Delta^T \tilde \Sigma_{\bG_i} \Delta  = \lambda_{\min} \Delta^T \Delta\]

where $\lambda_{\min}$ is the minimum eigenvalue of the normalized
covariance matrix $\tilde \Sigma_{\bG_{i}}$. Proceeding in a manner
similar to that in the proof of lemma \ref{lem:lb_scap_continuous}
we have that,

\[Pr\left(\frac{1}{n}d(\bX,\hat{\bX}(\bY)) \geq 2 d_0\right) \leq  \left(\dfrac{1}{1 + d_0 SNR \lambda_{\min}}\right)^{m/2} 2^{n (R_{X}(d_0) - K(n,d_0))} \]

From the above expression one can see that achievable sensing
capacity falls in general, since  $\lambda_{min} \leq 1$ as compared
to the case when the elements of $\bG$ are uncorrelated in which
case $\lambda_{\min} = 1 = \lambda_{\max}$.

\subsection{Deterministic $\bG$}
\label{subsec:deterministic}

In this section we will consider deterministic matrices $\bG$ and
provide upper bounds to sensing capacity for the general case. To
this end denote the rows of $\bG$ as $\bG_i,\,i=1,\,2,\,\ldots,\,m$.
Let the cross-correlations of these rows be denoted as:
$$
r_{i} = {\bG_{i}^{T} \bG_{i+1} \over \bG_{i}^{T} \bG_i}
$$
As before to ensure the SNR, to be fixed we impose $\bG_{i}^{T}
\bG_{i} = 1$ for all $i$. Then we have the following result:

\begin{lem}
For the generative models for the signal $\bX$ as outlined in the
problem set-up, an upper bound for the sensing capacity for a
deterministic sensing matrix $\bG \in \Real^{m\times n}$ is given
by:
\begin{equation}\label{lem:determ}
C(d_0) \leq \sum_{i=1}^{m-1} {\log \left (1 + SNR \alpha(1 - r_i)+
{r_i \alpha SNR \over \alpha SNR + 1}(1 + \alpha SNR (1-r_i)) \right
) \over R_X(d_0) - K(n,d_0)}
\end{equation}
\end{lem}

\begin{proof}
We will evaluate $I(\bX;\bY|\bG)$ via the straightforward method,

\[ I(\bX;\bY|\bG) = h(\bY|\bG) - h(\bY|\bG,\bX) \]

Note that $h(\bY|\bG,\bX) = h(\bN)$. Note that $h(\bY|\bG) \leq
h(\bY) \leq h(\bY^*)$ where $\bY^*$ is a Gaussian random vector
obtained via $\bG \bX^*$ where $\bX^*$ is now a Gaussian random
vector with i.i.d components and with the same covariance as $\bX$
under the generative model(s). We will now upper bound the entropy
of $\bY$ via,
\[
h(\bY) \leq h(\bY^*) \leq h(Y_1^*)+\sum_{i=1}^{m-1} h(Y_{i+1}^{*}
\mid Y_i^*) \leq h(Y_1^*) + h(Y_{i+1}^{*} - \eta_i Y_i^*)
\]
where $\eta_i Y_i^*$ is the best MMSE estimate for $Y_{i+1}^*$. The
MMSE estimate of $Y_{i+1}^*$ from $Y_i^*$ is given by,

\[ \hat{Y}_{i+1}^{*} = \frac{\Sigma_{Y_{i}^{*} Y_{i+1}^{*}}}{\Sigma_{Y_{i}^{*}}} Y_{i}^{*}\]

$\Sigma_{Y_{i}^{*}Y_{i+1}^{*}} = r_i \alpha SNR $ and
$\Sigma_{Y_{i}^{*}} =  \alpha SNR + 1 $. The result then follows by
evaluating the MMSE error given by,

\[ \ex (Y_{i+1}^{*} - \hat{Y}_{i+1}^{*})^2 = \ex \left( Y_{i+1}^{*} -
\frac{r_i \alpha SNR }{\alpha SNR + 1} Y_i^*\right)^2 \]

\[ \begin{array}{ll} \ex \left( Y_{i+1}^{*} -
\frac{r_i \alpha SNR }{\alpha SNR + 1} Y_i^* \right)^2  & = \alpha
SNR + 1 + \frac{(r_i \alpha SNR)^2 }{\alpha SNR + 1} - 2 \frac{(r_i
\alpha SNR)^2 }{\alpha SNR + 1}\\
& = 1 + \alpha SNR( 1 - r_i) + \frac{r_i\alpha SNR }{\alpha SNR + 1}
\left( 1 + (1 - r_i) \alpha SNR \right) \end{array} \]

Plugging in the quantities the result follows.

\end{proof}

Let us see the implications of the above result for one particular
type of sensing matrix architecture induced via a random filtering
and downsampling, considered in \cite{Wakin06}. The output of the
filter of length $L < n$ can be modeled via multiplication of $\bX$
via a Toeplitz matrix (with a banded structure). The overlap between
successive rows of the matrix $\bG$ is $L-1$ in this case implying a
large cross correlation $r_i$. From lemma \ref{lem:determ} it
follows that larger cross correlation in rows implies poor sensing
capacity. Also note that for a  filtering architecture one has to
address a coverage issue wherein it is required that $m > n - L +
1$. This implies that $L > n - m + 1$. Thus the filter length has to
be sufficiently large which implies that cross-correlation is also
large.

Indeed randomizing each row will lead to low cross-correlation (in
an expected sense) but the coverage issue still needs to be
addressed. On the other hand one can subsample the output signal of
length $ n - L + 1$ by some factor so as to reduce the cross
correlation yet ensuring coverage. In this case the matrix almost
becomes like a upper triangular matrix and there is a significant
loss of sensing diversity. A loose  tradeoff between the
filter-length $L$ and the sampling factor $d$ (say) immediately
follows from lemma \ref{lem:determ} where the cross correlation
changes according to $r_i = \dfrac{L(1-d)}{n}$

\begin{figure}[t]
\centering \makebox[0in]{
    \begin{tabular}{cc}
      \psfig{figure= 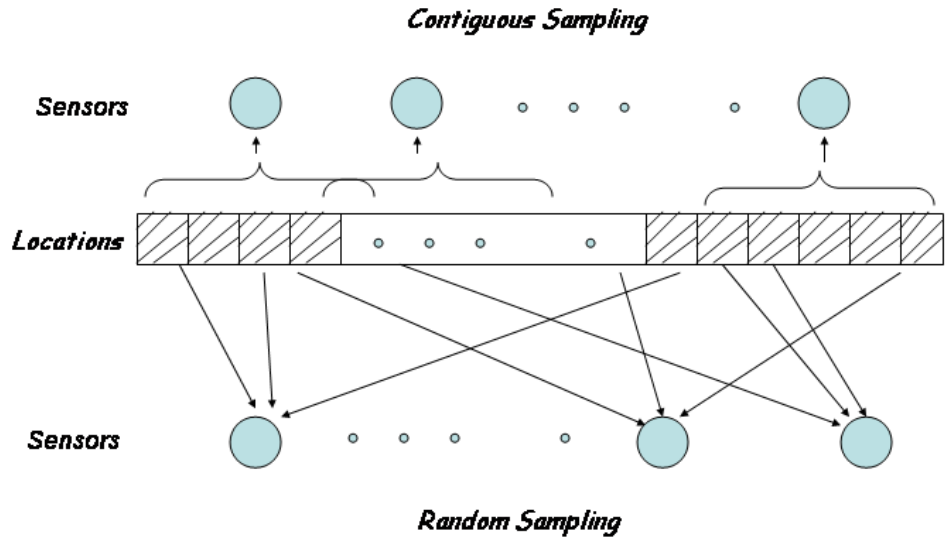,height= 4 in}
    \end{tabular}
  }
  \caption{Illustration of random sampling Vs contiguous sampling in a sensor network. This leads to different structures on the sensing
  matrix and that leads to different performance. }
  \label{fig:sampling}
\end{figure}

\section{Upper bounds on Sensing Capacity for $\left\{0,1\right\}$ ensemble}
\label{sec:ub_01ensemble}

The main motivation for considering this ensemble comes from
scenarios where randomization in the elements of $\bG$ is not
feasible, e.g. field estimation from smoothed data. In this case
each sensor measures a superposition of the signals that are in the
sensing range of the sensor. This leads us to consider other types
of modalities, e.g.  contiguous sampling of $\bX$ by each sensor Vs
random sampling for $\beta < 1$. An illustration of the two types of
sampling is shown in figure \ref{fig:sampling}. We reveal the
following contrast for the two cases for same $\beta < 1$

\begin{lem}
\label{lem:randsamp} Random Sampling: For the $\left\{0,1\right\}$
ensemble for sensing matrices consider the case when each row of
$\bG$ has $\beta n$ ones randomly placed in $n$ positions. Then for
discrete $\bX \in \left\{0,1\right\}^n$ drawn Bernoulli$(\alpha)$
and for $d_0 < \alpha$,

\[ C_{rand}(d_0) \leq \frac{H(J)}{h_2(\alpha) - h_2(d_0)} \]
where $H(.)$ is the discrete entropy function and where $J$ is a
random variable with distribution given by

\[\Pr(J=j)=\dfrac{\binom{\alpha n}{j}\binom{n(1-\alpha)}{\beta n - j}}{\binom{n}{\beta n}}\]

\end{lem}

\begin{proof}
See Appendix.
\end{proof}

\begin{lem}
\label{lem:contsamp} Contiguous Sampling: For the
$\left\{0,1\right\}$ ensemble for sensing matrices consider the case
where each row of $\bG$ has $\beta n$ \emph{consecutive} ones
randomly placed with wrap around. Then for discrete $\bX \in
\left\{0,1\right\}^n$ drawn Bernoulli$(\alpha)$ and $d_0 < \alpha$,

\[ C_{contg.}(d_0) \leq \frac{h_2(\alpha + \beta)}{h_2(\alpha) - h_2(d_0)} \]

\end{lem}

\begin{proof}
See Appendix.
\end{proof}

As seen the upper bound, $C_{rand}(d_0) \geq C_{contg.}(d_0)$. Thus
randomization in $\bG$ performs better.  The difference is shown in
figure \ref{fig:ub_randvscont} for a low sparsity scenario. The
proofs of the lemmas \ref{lem:randsamp} and \ref{lem:contsamp}
follow from the upper bounds to the mutual information terms as
provided in section \ref{sec:bound_info_01} and then applying the
necessary conditions for the lower bound on the probability of error
to be lower bounded by zero.

\begin{figure}[t]
\centering \makebox[0in]{
    \begin{tabular}{cc}
      \psfig{figure= 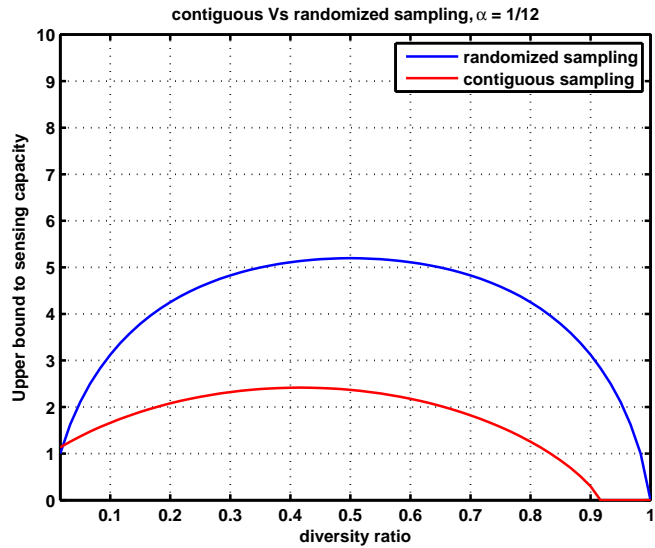,height= 3 in}
    \end{tabular}
  }
  \caption{A comparison of the upper bounds to sensing capacity for the randomized sampling Vs contiguous sampling case.
  $\bX$ is the Bernoulli model and the ensemble for $\bG$ is the $\left\{0,1\right\}$ ensemble. We have selected the case of
  low sparsity in this case. Note that due to loose overbounding  of mutual information (we basically got rid of noise)
  the upper bounds are greater than in the case of Gaussian ensemble. }
  \label{fig:ub_randvscont}
\end{figure}

\section{Estimation of functions of $\bX$}
\label{sec:func_est}

The analysis of lower bounds to the probability of error presented
in this paper extend in a straightforward way to estimation of
functions of $\bX$. In this section we will consider one such
scenario that has received attention in  relation to problems
arising in physics. The discussion below will reveal the power of
the method presented in this work and it is easily capable of
handling more complicated cases and scenarios, though the
computation of the terms involved in the analysis may become hard.

\subsection{Detecting the sign pattern of $\bX$}
Of particular interest is to estimate the sign pattern of the
underlying signal $\bX$. To this end define a new random variable
$\bU$, via

\[ U_i = \left\{ \begin{array}{ll}
                1 & \mbox{if} \,\, X_i > 0 \\
                -1 & \mbox{if} \,\, X_i < 0 \\
                0  & \mbox{if} \,\, X_i = 0
                \end{array} \right. \]

The corresponding $n$ dimensional extension and probability
distribution on $\bU$ is induced directly via $P_{\bX}$. In such a
case note that $\bU \rightarrow \bX \rightarrow \bY \rightarrow
\hat{\bU}(\bY)$ forms a Markov chain. To this end consider an error
event defined via,

\[ E =
\left\{\begin{array}{ll}
 1  & \mbox{if} \,\, \bU \neq \hat{\bU}(\bY) \\
 0  & \mbox{otherwise}
  \end{array}\right. \]

Then we have,

\[ \begin{array}{lll}
H(\bU,E| \bY) & = & \underset{\leq 1}{\underbrace{H(E|\bY)}} + H(\bU|E,\bY) \\
             & = & H(\bU|\bY) + \underset{=0}{\underbrace{H(E|\bU,\bY)}}
\end{array} \]

Thus we have

\[ H(\bU|\bY) \leq 1 + P_e \underset{\leq n \log 3}{\underbrace{H(\bU|E =1, \bY)}} + \underset{=0}{\underbrace{(1 - P_e) H(\bU| E =
0,\bY)}} \] This implies,

\[ P_e \geq \frac{H(\bU) - I(\bU;\bY|\bG) -1}{n\log 3} \]

In order to evaluate the $I(\bU;\bY|\bG)$ we note that
$I(\bU,\bX;\bY|\bG) = I(\bX;\bY|\bG)$. This follows from ,
$I(\bU,\bX;\bY|\bG) = H(\bU,\bX) - H(\bX,\bU|\bY,\bG) = H(\bX) -
H(\bX|\bG,\bY) - H(\bU|\bG,\bY,\bX) = I(\bX;\bY|\bG)$. Thus
$I(\bU;\bY|\bG) = I(\bX;\bY|\bG) - I(\bX;\bY|\bG,\bU)$ and both
these terms can be adequately bounded/evaluated.

\section{Appendix}
\subsection{Proof of lemma \ref{lem:lowerbound}}
Let $X^{n} = \left\{X_1,...,X_n\right\}$ be an i.i.d. sequence where
each variable $X_i$ is distributed according to a distribution
$P_{X}$ defined on the alphabet ${\cal X}$. Denote $P_{X^n}
\triangleq (P_{X})^n$ the n-dimensional distribution induced by
$P_{X}$. Let the space ${\cal X}^{n}$ be equipped with a distance
measure $d(.,.)$ with the distance in $n$ dimensions given by
$d_n(X^n,Z^{n}) = \sum_{k=1}^{n} d(X_k,Z_k)$ for $X^n, Z^n \in {\cal
X}^n$. Given $\epsilon >0$, there exist a set of points
$\left\{Z_{1}^{n},...,Z_{N_{\epsilon}(n,d_0)}\right\} \subset {\cal
X}^n$ such that,

\begin{eqnarray}
\label{eq:cover} P_{X^n}\left(\bigcup_{i=1}^{N_{\epsilon}(n,d_0)}
{\cal B}_{i}\right) \geq 1 - \epsilon
\end{eqnarray}

where ${\cal B}_{i} \triangleq \left\{ X^n : \frac{1}{n}
d_n(X^n,Z_{i}^{n}) \leq d_0 \right\}$, i.e.,  the $d_0$ balls around
the set of points \emph{cover} the space ${\cal X}^n$ in probability
exceeding $1 - \epsilon$.

Given such set of points there exists a function $f(X^n): X^n
\rightarrow Z_{i}^{n}\,\, s.t. \,\, Pr\left(\frac{1}{n}
d_n(X^n,Z_{i}^{n}) \leq d_0\right) \geq 1 - \epsilon$. To this end,
let $ T_{P_{{X}^{n}}}$ denote the set of $\delta$ - typical
sequences in ${\cal X}^n$ that are typical $P_{X^n}$, i.e.

\[ T_{P_{{X}^{n}}} = \left\{ X^n : | -\frac{1}{n} \log \hat{P}(X^n) - H(X)| \leq \delta \right\} \]

where $\hat{P}(X^n)$ is the empirical distribution induced by the
sequence $X^n$. We have the following lemma from \cite{cover}.

\begin{lem}
\label{lem:AEP} For any $\eta > 0$ there exists an $n_0$ such that
for all $n \geq n_0$, such that

\[ Pr\left( X^n : |-\frac{1}{n} \log \hat{P}(X^n) - H(X)| < \delta \right) > 1 -
\eta \]
\end{lem}
In the following we choose $\eta = \delta$. Given that there is an
algorithm $\hat{X}^n(\bY)$ that produces an estimate of $X^n$ given
the observation $\bY$. To this end define an error event on the
algorithm as follows,

\[ E_n =
\left\{\begin{array}{l}
 1  \,\,\mbox{if} \,\,\, \frac{1}{n}d_n(X^n,\hat{X}^n(\bY)) \geq d_0 \\
 0  \,\, \mbox{otherwise}
  \end{array}\right. \]

Define another event $A_n$ as follows

\[ A_n =
\left\{\begin{array}{l}
 1 \,\,\mbox{if} \,\,\, X^n \in T_{P_{X^n}} \\
 0  \,\, \mbox{otherwise}
  \end{array}\right. \]

Note that  since $X^n$ is drawn according to $P_{X^n}$ and given
$\delta > 0$ we choose $n_0$ such that conditions of lemma
\ref{lem:AEP} are satisfied. In the following we choose $n \geq
n_0(\delta)$. Then a priori, $Pr(A_n = 1) \geq (1 - \delta)$. Now,
consider the following expansion,

\[\begin{array}{l}
H(f(X^n),E_n, A_n|\bY) \\
=  H(f(X^n)|\bY) + H(E_n, A_n | f(X^n),\bY) \\
=  H(E_n , A_n |\bY) + H(f(X^n)|E_n, A_n, \bY)
\end{array} \]

This implies that

\[ \begin{array}{l}
H(f(X^n)|\bY)\\
= H(E_n, A_n |\bY) - H(E_n , A_n |f(X^n),\bY) + H(f(X^n)|E_n,A_n,\bY) \\
= I(E_n , A_n ; f(X^n)|\bY) + H(f(X^n)|E_n, A_n ,\bY) \\
\leq H(E_n, A_n)  + H(f(X^n)|E_n, A_n, \bY)\\
\leq H(E_n) + H(A_n) + H(f(X^n)|E_n,A_n,\bY)
\end{array}\]

Note that $H(E_n) \leq 1$ and $H(A_n) = \delta \log \frac{1}{\delta}
+ (1 - \delta) \log\frac{1}{1 - \delta} \sim \delta $. Thus we have

\[ \begin{array}{ll}
 & H(f(X^n)|\bY) \leq 1 + \delta + P_{e}^{n} H(f(X^n)|\bY,E_n = 1, A_n ) \\ & + (1
-P_{e}^{n})H(f(X^n)|\bY,E_n = 0, A_n) \end{array} \]

Now the term $P_{e}^{n} H(f(X^n)|\bY,E_n = 1, A_n) \leq P_{e}^{n}
\log N_{\epsilon}(n,d_0) $. Note that the second term  does not go
to zero. For the second term we have that,


\[\begin{array}{l} (1 -P_{e}^{n})H(f(X^n)|\bY,E_n = 0, A_n)  \\
= P(A_n =1) (1- P_{e}^{n}) H(f(X^n) | \bY, E_n = 0, A_n=1) \\
\hspace{5mm} + P(A_n = 0
)(1- P_{e}^{n}) H(f(X^n) | \bY, E_n = 0, A_n=0) \\
\leq (1 - P_{e}^{n}) H(f(X^n) | \bY, E_n = 0, A_n=1) \\
\hspace{5mm}  + \delta (1 - P_{e}^{n})
\log\left(N_{\epsilon}(n,d_0)\right) \end{array}\]

The first term on R.H.S in the above inequality is bounded via,

\[ (1 -P_{e}^{n})H(f(X^n)|\bY,E_n = 0, A_n =1 ) \leq (1- P_{e}^{n}) \log \left( |{\cal S}| \right) \]

where ${\cal S}$ is the set given by,

\[ {\cal S} = \left\{ i : d_{set}\left({{\cal B}_{f(X^n)}, \cal B}_{i}\right)
\leq d_0 \right\} \]

where $d_{set}(S_1,S_2) = \min_{s \in S_1, s' \in S_2}d_{n}(s,s')$
is the set distance between two sets. Now note that $I(f(X^n);X^n) =
H(f(X^n))$ and $H(f(X^n)|\bY) = H(f(X^n)) - I(f(X^n);X^n) \geq
H(f(X^n)) - I(X^n;\bY)$ where the second inequality follows from
data processing inequality over the Markov chain $f(X^n)
\leftrightarrow X^n \leftrightarrow \bY$. Thus we have,

\[ \begin{array}{ll} P_{e}^{n} \geq &  \dfrac{I(f(X^n);X^n) - \log |{\cal S}| - I(X^n;\bY) - 1}{(1 - \delta) \log N_{\epsilon}(n,d_0) - \log |{\cal
S}|} \\\\\ & - \dfrac{\delta ( 1 + \log N_{\epsilon}(n,d_0))}{(1 -
\delta)\log N_{\epsilon}(n,d_0) - \log|{\cal S}|}
\end{array} \]

The above inequality is true for all the mappings $f$ satisfying the
distortion criteria for mapping $X^n$  and for all choices of the
set satisfying the covering condition given by \ref{eq:cover}. We
now state the following lemma for a minimal covering, taken from
\cite{csiszar}.

\begin{lem}
\label{lem:mincover} Given $\epsilon> 0$ and the distortion measure
$d_n(.,.)$, let $N_{\epsilon}(n,d_0)$ be the minimal number of
points $Z_{1}^{n},..., Z_{N_{\epsilon}(n,d_0)}^{n} \subset {\cal
X}^n$ satisfying the covering condition,
\[
\label{eq:cover}
 P_{X^n}\left(\bigcup_{i=1}^{N_{\epsilon}(n,d_0)} {\cal B}_{i}\right) \geq 1 - \epsilon
 \]
Let $N_{\epsilon}(n,d_0)$ be the minimal such number. Then,

\[ \limsup_n \frac{1}{n} N_{\epsilon}(n,d_0) = R_{X}(\epsilon, d_0) \]

where $R_{X}(\epsilon, d_0)$ is the infimum of the $\epsilon$-
achievable rates at distortion level $d_0$.
\end{lem}

Note that $\lim_{\epsilon \downarrow 0} R_{X}(\epsilon, d_0) =
R_{X}(d_0)$ where $R_{X}(d_0)= \min_{p(\hat{X}|X)} I(\hat{X};X)
\,\,\, \mbox{subject to} \,\, \frac{1}{n} E(d(X^n,\hat{X}^n)) \leq
d_0$. In order to lower bound $P_{e}^{n}$ we choose the mapping
$f(X^n)$ to correspond to the minimal cover. Also w.l.o.g we choose
$\delta = \epsilon$. We note the following.

\begin{enumerate}
\item From lemma \ref{lem:AEP}, given $\epsilon >0$,
$\exists n_0(\epsilon)$ such that for all $n \geq n_0(\epsilon)$, we
have $Pr(T_{P_{X^n}}) \geq 1 - \epsilon$.

\item Given $\epsilon > 0$ and for all $\beta > 0$, for the minimal cover we have from lemma \ref{lem:mincover} that
$\exists \, n_1(\beta)$ such that for all $n \geq n_1(\beta)$,
$N_{\epsilon}(n,d_0) \leq n (R_{X}(\epsilon, d_0) + \beta )$.

\item From the definition of the rate distortion function we have
for the choice of the functions $f(X^n)$ that satisfies the
distortion criteria,  $ I(f(X^n);X^n) \geq n R_X(\epsilon, d_0)$.
\end{enumerate}

Therefore we have for $n \geq \max(n_0,n_1)$,

\[ \begin{array}{ll} P_{e}^{n} \geq &  \dfrac{n R_{X}(\epsilon, d_0) - \log |{\cal S}| - I(X^n;\bY) - 1}{(1 - \epsilon)( n
(R_{X}(\epsilon, d_0) + \beta) - \log |{\cal S}|} \\\\ & -
\dfrac{\epsilon ( 1 + n (R_{X}(\epsilon, d_0) + \beta)}{(1 -
\epsilon) n (R_{X}(\epsilon, d_0) + \beta) - \log|{\cal S}|}
\end{array}\]


Clearly, $\log |S| \leq \frac{n}{2} R_{X}(\epsilon, d_0)$. 

\paragraph{Limiting case} Since the choice of $\epsilon,\beta$ is arbitrary we can choose them to be arbitrary
small. In fact we can choose $\epsilon, \beta
\downarrow 0$. Also note that for every $\epsilon > 0$ and $\beta
> 0$ there exists $n_2(\beta)$ such that $R_{X}(d_0) + \beta \geq R_{X}(\epsilon, d_0) \geq
R_{X}(d_0) - \beta$. Therefore for all $n \geq \max(n_0,n_1,n_2)$
in the limiting case when $\epsilon,\beta \downarrow 0$, we have

\[ P_{e} \geq \frac{R_{X}(d_0) - \frac{1}{n}\log |{\cal S}| - \frac{1}{n} I(X^n;\bY) }{R_{X}(d_0)- \frac{1}{n} \log |{\cal S}|} - o(1) \]

This implies that

\[ P_{e} \geq \frac{R_{X}(d_0) - \frac{1}{n}\log |{\cal S}| - \frac{1}{n} I(X^n;\bY) }{R_{X}(d_0)} - o(1) \]

The proof then follows by identifying $K(n,d_0) = \frac{1}{n} \log
|{\cal S}|$, and is bounded above by a constant.

\subsection{Proof of lemma \ref{lem:discrete_LB} }

\begin{proof}
Given an observation $\bY$ about the event $X^n$. Define an error
event,

\[ E =
\left\{\begin{array}{l}
 1  \,\,\mbox{if} \,\,\, \frac{1}{n}d_{H}(X^n,\hat{X}^n(\bY)) \geq d_0 \\
 0  \,\, \mbox{otherwise}
  \end{array}\right. \]

Expanding $H(X^n,E|\bY)$ in two different ways we get that,

\[ H(X^n|\bY) \leq 1 + n P_e \log (|{\cal X}|) + (1 - P_e) H(X^n| E =
0,\bY) \]

Now the term
\[ \begin{array}{l} ( 1 - P_e) H(X^n | E = 0, \bY)\\
 \leq (1 - P_e)
\binom{n}{d_0 n}(|{\cal X}| -1)^{nd_0}  \\
\leq n (1- P_e) \left(h(d_0) + d_0 \log (|{\cal X}| -1) \right)
\end{array} \]

Then we have for the lower bound on the probability of error that,

\[ P_e \geq \frac{ H(X^n|\bY) - n \left(h(d_0) + d_0 \log (|{\cal X}|
-1)\right) ) - 1}{n \log(|{\cal X}|) - n \left(h(d_0) + d_0 \log
(|{\cal X}|-1)\right) } \]

Since $H(X^n|\bY) = H(X^n) - I(X^n;\bY)$ we have

\[ P_e \geq \frac{n\left(H(X)  - h(d_0) - d_0 \log (|{\cal X}|
-1)\right) - I(X^n;\bY) - 1}{n \log(|{\cal X}|) - n \left(h(d_0) +
d_0 \log (|{\cal X}|-1)\right) } \]

It is known that $R_{X}(d_0) \geq  H(X)  - h(d_0) - d_0 \log (|{\cal
X}| -1) $, with equality iff \[ d_0 \leq (|{\cal X}|-1)\min_{X \in
{\cal X}} P_{X} \] see e.g., \cite{csiszar}. Thus for those values
of distortion we have for all $n$,

\[ P_e \geq \frac{nR_{X}(d_0) - I(X^n;\bY) - 1}{n \log(|{\cal X}|) - n \left(h(d_0) +
d_0 \log (|{\cal X}|-1)\right)}  \]
\end{proof}

\subsection{Rate distortion function for the mixture Gaussian source
under squared distortion measure} \label{subsec:ratedist_gauss}

It has been shown in \cite{zamirIT02} that the rate distortion
function for a mixture of two Gaussian sources with variances given
by $\sigma_1$ with mixture ratio $\alpha$ and $\sigma_0$ with
mixture ratio $1- \alpha$, is given by

\[ \begin{array} {l} R_{mix}(D) = \\ \left\{ \begin{array}{l}
H(\alpha) + \frac{(1 - \alpha)}{2} \log(\frac{\sigma_{0}^{2}}{D}) + \frac{\alpha}{2}
\log (\frac{\sigma_{1}^{2}}{D}) \,\, \mbox{if} \,\, D < \sigma_{0}^{2} \\
H(\alpha) + \frac{\alpha}{2} \log(\frac{\alpha \sigma_{1}^{2}}{D -
(1 - \alpha)\sigma_{0}^{2}}) \,\, \mbox{if} \,\, \sigma_{0}^{2} < D
\leq (1 - \alpha)\sigma_{0}^{2} + \alpha \sigma_{1}^{2} \end{array}
\right. \end{array}\]

For a strict sparsity model we have $\sigma_{0}^2 \rightarrow 0$ we
have that,

\[ R_{mix}(D) = \begin{array}{l} H(\alpha) + \frac{\alpha}{2}
\log(\frac{\alpha \sigma_{1}^{2}}{D} ) \,\, \mbox{if}
\,\, 0 < D \leq \alpha \sigma_{1}^{2} \end{array}  \]

\subsection{Bounds on Mutual information}

In this section we will evaluate bounds on mutual information that
will be useful in characterization of the Sensing Capacity. Given
that the matrix $\bG$ is chosen independently of $\bX$ we expand the
mutual information between $\bX$ and $\bY,\bG$ in two different ways
as follows --

\[\begin{array}{ll}
 I(\bX;\bY,\bG) & = \underset{=0}{\underbrace{I(\bX;\bG)}} + I(\bX;\bY|\bG)\\
                & =  I(\bX;\bY) + I(\bX;\bG|\bY) \end{array}\]

This way of expanding gives us handle onto evaluating the mutual
information with respect to the structure of the resulting sensing
matrix $\bG$. From above we get that,

\[ \begin{array}{ll}
I(\bX;\bY|\bG)  & = I(\bX;\bY) + I(\bX;\bG|\bY)\\
                & = h(\bY) - h(\bY|\bX) + h(\bG|\bY) - h(\bG|\bX,\bY) \end{array}\]

To this end we have the following lemma.

\begin{lem}
\label{lem:bound_info1} For a sparsity level of $\alpha$ and
diversity factor of $\beta = 1$,
\[ I(\bX;\bY|\bG) \leq \frac{m}{2} \log( 1 + \frac{\alpha P}{N_0})
\]
\end{lem}

\begin{proof}
First note that,

\[h(\bY) \leq \frac{m}{2} \log 2\pi e(N_0 + \alpha P) \]

Since conditioned on $\bX$, $\bY$ is distributed with a Gaussian
density we have,

\[h(\bY|\bX)  = \frac{m}{2} \log 2 \pi e \left(N_0 + \frac{\sum_{i=1}^{k} \bX_{i}^{2} P}{n}\right)
\]

\[h(\bG|\bY) \leq h(\bG) = \frac{mn}{2}\log\left(2\pi e\frac{P}{n}\right) \]

Note also that conditioned on $\bX$ and $\bY$ the $\bG$ has a
Gaussian distribution. Now note that, $h(\bG|\bY,\bX)$. First note
that, rows of $\bG$ are independent of each other given $\bX$ and
$\bY$. So we can write,

\[ h(\bG|\bY,\bX) = m h(\mathbf{g_1} | \bY,\bX) \]
where $\mathbf{g}_1$ is the first row of the matrix $\bG$. Since
$\mathbf{g}$ is Gaussian one can find the residual entropy in terms
of the residual MMSE error in estimation of $\mathbf{g}$ given $\bX$
and $\bY$. This error is given by --

\[ \begin{array}{ll} \mbox{MMSE}_{\mathbf{g}_1|\bY,\bX} & = \Sigma_{\mathbf{g}_1|\bX} -
\Sigma_{\mathbf{g}_1 \bY|\bX} \Sigma_{\bY|\bX}^{-1}
\Sigma_{\mathbf{g}_1 \bY|\bX}^{T} \\
& = \Sigma_{\mathbf{g}_1} - \Sigma_{\mathbf{g}_1\bY_1|\bX}
\Sigma_{\bY_1|\bX}^{-1} \Sigma_{\mathbf{g}_1\bY_1|\bX}^{T}
\end{array} \]

The second equation follows from the fact that $\bG$ is independent
of $\bX$ and given $\bX$ the row $\mathbf{g}_1$ is independent of
other observations, $\bY_2,...,\bY_m$. First note that given $\bX$
we also know which positions of $\bX$ are zeros. So without lossof
generality we can assume that the first $k$ elements of $\bX$ are
non-zeros and the rest are zeros. Now note the following,

\[ \Sigma_{\mathbf{g}_1} = \frac{P}{n} I_{n} \]
\[ \Sigma_{\mathbf{g}_{1} \bY_1 |\bX} = \frac{P}{n}\begin{pmatrix}
                                          \bX_{1} \\
                                          \vdots\\
                                          \bX_{k}\\
                                          \mathbf{0}_{n-k}
                                        \end{pmatrix} \]
where $\mathbf{0}_{n-k}$ is a column vector of $n-k$ zeros.

\[ \Sigma_{\bY_1|\bX} = \frac{P}{n} \sum_{i=1}^{k} \bX_{i}^{2} + N_0
\]

Therefore we have,

\[ \begin{array}{l} h(\mathbf{g}_1|\bY_1,\bX) \\ = \frac{1}{2}\log (2\pi e)^{k} \det \left(\frac{P}{n}I_k
- \frac{P}{n} \bX_{1:k} \Sigma_{\bY_1|\bX}^{-1} \frac{P}{n}
\bX_{1:k}^{T} \right) \\
+ \frac{n-k}{2} \log 2\pi e \frac{P}{n}
\end{array}
\]

Note that the second term on the R.H.S in the above equation
corresponds to the entropy of those elements of the row
$\mathbf{g}_1$ that have no correlation with $\bY$, i.e. nothing can
be inferred about these elements since they overlap with zero
elements of $\bX$.  Now, using the equation $\det(I + AB) = \det ( I
+ BA)$, we have that

\[ \begin{array}{ll} h(\mathbf{g}_1|\bY_1,\bX) & = \frac{1}{2} \log (\frac{2\pi e P}{n})^{k} \det
\left(1 -  \bX_{1:k}^{T} \Sigma_{\bY_1|\bX}^{-1} \frac{P}{n}
\bX_{1:k} \right) \\
& = \frac{1}{2}\log \left( (\frac{2 \pi e P}{n})^k
\frac{N_0}{\frac{P}{n} \sum_{i=1}^{k} \bX_{i}^{2} +
N_0}\right)\end{array}\]

Plugging in all the expressions we get a lower bound on  the  mutual
information $I( \bX;\bY|\bG)$ -

\[ I(\bX;\bY|\bG) \leq \frac{m}{2} \log(1 + \frac{\alpha P}{N_0} )
\]

\end{proof}

In contrast to the upper bound derived in the proof of lemmas
\ref{lem:scap_binom} and \ref{lem:scap_continuous}, this alternate
derivation provides a handle to understand the effect of the
structure of $\bG$ on the mutual information when one is not allowed
to pick a maximizing input distribution on $\bX$. Moreover the above
derivation can potentially handle scenarios of correlated $\bG$.
Below we will use the above result in order to prove lemma
\ref{lem:bound_info2}.

%

\subsection{Proof of lemma \ref{lem:bound_info2}}

To this end let $l = \beta n$ and is fixed, i.e. there are only $l$
non-zero terms in each row of matrix $\bG$. We have
\[ h(\bG) = \frac{ml}{2} \log 2\pi e \frac{P}{l} + m h_2(\beta) \]

Now we will first evaluate $h(\bG|\bY,\bX)$. Proceeding as in
derivation of lemma \ref{lem:bound_info1}, we have that,

\[ h(\bG|\bX,\bY) = m h(\mathbf{g}_1|\bY_1,\bX) +  m h_2(\beta) \]

where one can see that if the matrix $\bG$ is chosen from a Gaussian
ensemble then given $\bX$ and $\bY$ it tells nothing about the
positions of the non-zeros in each row. Hence the additive term
$h_2(\beta)$ appears in both terms and is thus canceled in the
overall calculations. So we will omit this term in the subsequent
calculations. To this end, let $j$ denote the number of overlaps of
the vector $\mathbf{g}_1$ and the k-sparse vector $\bX$. Given
$\bY_1$ and $\bX$ one can only infer something about those elements
of $\bG$ that contribute to $\bY_1$. Given the number of overlaps
$j$ we then have

\[\begin{array}{l} h(\mathbf{g}_1|\bX,\bY_1,j)  = \frac{l-j}{2} \log 2\pi e \frac{P}{l}
+ \frac{1}{2} \log \left((\frac{2\pi e P}{l})^{j}
\frac{N_0}{\frac{P}{l}\sum_{i=1}^{j} \bX_{j}^{2} + N_0}\right)
\end{array}\]

where we have assumed without loss of generality that the first $j$
elements of $\bX$ are non-zero and overlap with elements of the
first row. Now note that,

\[ h(\bY|j) \leq \frac{m}{2} \log 2\pi e (\frac{P j}{l} + N_0) \]
\[h(\bY|\bX,j) = \frac{m}{2} \log 2 \pi e \left(\frac{P}{l}
\sum_{i=1}^{j} \bX_{i}^{2} + N_0 \right) \]

From above we have that,

\[ I(\bX;\bY|\bG,j) = \frac{m}{2} \log(1 + \frac{j P}{l N_0}) \]

Taking the expectation with respect to the variable $j$ we have,

\[ I(\bX;\bY|\bG) = \frac{m}{2} \ex_{j} \log(1 + \frac{j P}{l N_0})
\]

Note that $j \leq \min \left\{k,l\right\}$ and has a distribution
given by,

\[\Pr(j)=\dfrac{\binom{k}{j}\binom{n-k}{l-j}}{\binom{n}{l}}\]

\section{Upper bounds to Mutual information for $\left\{0,1\right\}$
ensemble} \label{sec:bound_info_01}

In this section we will derive upper bounds to the mutual
information $I(\bX;\bY|\bG)$ for the case when the matrix is chosen
from a $\left\{0,1\right\}$ ensemble. First it is easily seen that
for this ensemble a full diversity leads to loss of rank and thus
the mutual information is close to zero. So we will only consider
the case $\beta < 1$.

\subsection{Random locations of $1$'s in $\bG$}

In this section we will provide simple upper bounds to the mutual
information $I(\bX;\bY|\bG)$ for the case of $\left\{0,1\right\}$
ensemble of sensing matrices. Note that,

\[ I(\bX;\bY|\bG) \leq I(\bX;\bG\bX|\bG)\]

Let $\tilde{\bY} = \bG\bX$. Then we have,

\[ I(\bX;\tilde \bY|\bG) = I(\bX;\tilde \bY) + I(\bX;\bG | \tilde
\bY) \]

Now note that $\frac{1}{n} I(\bX;\tilde \bY) = o(1)$. Then we need
to evaluate $I(\bG;\bX|\tilde \bY) \leq H(\bG) - H(\bG|\tilde
\bY,\bX)$. Now note that since each row of $\bG$ is an independent
Bernoulli$\sim \beta$ sequence we can split the entropy into sum of
entropies each individual rows. To this end focus on the first row.
Then conditioned on there being $l$ $1$'s in the row we have,

$H(\bG_1|l) \leq \binom{n}{l}$. Given that $X$ is $k$-sparse we
have,

\[ H(\bG_1|\bX,\tilde \bY,l,k) = \sum_{j=0}^{\min(k,l)} \dfrac{\binom{k}{j}
\binom{n-k}{l -j}}{\binom {n}{l}} \log \binom{k}{j} \binom{n-k}{l
-j}\]

Thus we have
\[ I(\bX;\bG|\tilde \bY, k, l) \leq \binom{n}{l} - \sum_{j=0}^{\min(k,l)} \dfrac{\binom{k}{j}
\binom{n-k}{l -j}}{\binom {n}{l}} \log \binom{k}{j} \binom{n-k}{l
-j}  = H(J | k,l)\]

where $J$ is a random variable with distribution given by,

\[ Pr(J = j) = \dfrac{\binom{k}{j}
\binom{n-k}{l -j}}{\binom {n}{l}} \]

For large enough $n$, $k = \alpha n$ and $l = \beta n$ w.h.p. Thus
$I(\bX;\bG|\bY) \leq H(\tilde J)$, where $\tilde J$ has a limiting
distribution given by,

\[ Pr(\tilde J = j) = \dfrac{\binom{\alpha n}{j}
\binom{n(1 - \alpha)}{\beta n -j}}{\binom {n}{\beta n}} \]

In other words given $\epsilon >0$ there exists an $n_0$ such that
for all $n \geq n_0$, $\underset{j}{\sup} |P_{J}(j) - P_{\tilde
J}(j)| \leq \epsilon$ and by continuity of the entropy function,
[\cite{csiszar}, pp. 33, Lemma 2.7], it follows that $|H(J) -
H(\tilde J)| \leq -\epsilon \log \dfrac{\epsilon}{n}$

\subsection{Contiguous sampling}

In this case for each row we have $H(\bG_1) = \log n$. To evaluate
$H(\bG_1|\bX,\tilde \bY)$, fix the number of ones in $\bG_1$ to be
equal to $l$ and the number of non-zero elements in $\bX$ to be
equal to $k$. Now note that if $\tilde Y_1 = 0$ then there is no
overlap in $\bG_1$ and $\bX$. This means that the row of $\bG$ can
have contiguous ones in $n-k -l$ positions equally likely. The
probability of no overlap is $\frac{n-k-l}{n}$. On the other hand if
$\tilde Y_1 > 0$, then uncertainty in locations of ones in $\bG_1$
reduces to $\log(k+l)$. The probability that $Y > 0$ is $\frac{k +
l}{n}$. Thus we have,

\[ I(\bG_1;\bX|\tilde \bY) \leq  m H(O)\]

where $O$ is a binary random variable with distribution $(1 -
\frac{k+l}{n}), \frac{k+l}{n}$. For large enough $n$ this comes
close to $1 - (\alpha + \beta), \alpha + \beta$. Thus we have,

\[ I(\bG;\bX|\bY) \leq  m H(\alpha + \beta)\]

\bibliographystyle{IEEEtran}
\bibliography{CS_master1}

\begin{thebibliography}{10}
\providecommand{\url}[1]{#1}
\csname url@rmstyle\endcsname
\providecommand{\newblock}{\relax}
\providecommand{\bibinfo}[2]{#2}
\providecommand\BIBentrySTDinterwordspacing{\spaceskip=0pt\relax}
\providecommand\BIBentryALTinterwordstretchfactor{4}
\providecommand\BIBentryALTinterwordspacing{\spaceskip=\fontdimen2\font plus
\BIBentryALTinterwordstretchfactor\fontdimen3\font minus
  \fontdimen4\font\relax}
\providecommand\BIBforeignlanguage[2]{{%
\expandafter\ifx\csname l@#1\endcsname\relax
\typeout{** WARNING: IEEEtran.bst: No hyphenation pattern has been}%
\typeout{** loaded for the language `#1'. Using the pattern for}%
\typeout{** the default language instead.}%
\else
\language=\csname l@#1\endcsname
\fi
#2}}

\bibitem{heroSSP05}
R.~Rangarajan, R.~Raich, and A.~Hero, ``Sequential design of experiments for a
  rayleigh inverse scattering problem,'' in \emph{IEEE Workshop on Statistical
  Signal Processing (SSP)}, Bordeaux, France, July 2005.

\bibitem{mimoradar}
Y.~Yang and R.~Blum, ``Waveform design for mimo radar based on mutual
  information and minimum mean-square error estimation,'' ser. Conference on
  Information System and Sceinces.

\bibitem{Donoho1}
D.~Donoho, ``Compressed sensing,'' \emph{IEEE Transactions on Information
  Theory}, vol.~52, no.~4, pp. 1289--1306, April 2006.

\bibitem{Candes1}
E.~Candes and T.~Tao, ``Near optimal signal recovery from random projections:
  Universal encoding strategies?'' \emph{preprint}, 2004.

\bibitem{Nowak2}
M.~Rabbat, J.~Haupt, A.~Singh, and R.~Nowak, ``Decentralized compression and
  predistribution via randomized gossiping,'' ser. International Conference on
  Information Processing in Sensor Networks, Nashville, TN, USA, April 2006.

\bibitem{Nowak3}
W.~Bajwa, J.~Haupt, A.~Sayeed, and R.~Nowak, ``Compressive wireless sensing,''
  ser. International Conference on Information Processing in Sensor Networks,
  Nashville, TN, USA, April 2006.

\bibitem{Konrad05}
M.~Mole, P.~Ward, I.~Hochman, K.~Lopez, J.~Konrad, and W.~Karl, ``ipark -
  vison-based parking monitoring system,'' 5th annual IEEE Student Design
  Contest at Rochester Institute of Technology (RIT), 2005.

\bibitem{McEliece}
J.~T.~R. McEliece, ``Data fusion algorithms for collaborative robotic
  exploration,'' in \emph{The Interplanetary Network Progress Report, IPN PR
  42-149}, Jan-March 2002, pp. 1--14.

\bibitem{Nowak1}
J.~Haupt and R.~Nowak, ``Signal reconstruction from noisy random projections,''
  \emph{IEEE Transactions on Inforamtion Theory}, vol.~52, no.~9, pp.
  4036--4068, Sep 2006.

\bibitem{srv_murat1}
O.~Savas, M.~Alanyali, and V.~Saligrama, ``Randomized sequential algorithms for
  data aggregation in sensor networks,'' ser. Conference on Information System
  and Sciences, Princeton, NJ, USA, 2006.

\bibitem{Rachlin1}
Y.~Rachlin, R.~Negi, and P.~Khosla, ``Sensing capacity for target detection,''
  ser. Information Theory Workshop, 2004.

\bibitem{Boyd}
S.~Boyd and L.~Vandenberghe, \emph{Convex Optimization}.\hskip 1em plus 0.5em
  minus 0.4em\relax Cambridge University Press, 2006.

\bibitem{Tibshirani96}
R.~Tibshirani, ``Regression shrinkage and selection via the lasso,''
  \emph{Journal of the Royal Statistical Society}, vol.~58, no.~1, pp.
  267--288, April 1996.

\bibitem{knight00}
K.~Knight and W.~Fu, ``Asymptotics for lasso-type estimators,'' \emph{The
  Annals of Statistics}, vol.~28, no.~5, pp. 1356--1378, Oct 2000.

\bibitem{Fan01}
J.~Fan and R.~Li, ``Variable selection via nonconcave penalized likelihood and
  its oracle properties,'' \emph{Journal of the American Statistical
  Association}, vol.~96, no. 456, pp. 1138--1360, Dec 2001.

\bibitem{csiszar}
I.~Csisz$\mathrm{\acute{a}}$r and J.~J. Korner, \emph{Information Theory:
  Coding Theorems for Discrete Memoryless Systems}, ser. Academic Press, New
  York, 1981.

\bibitem{zegerIT94}
K.~Zeger and A.~Gersho, ``Number of nearest neighbors in a euclidean code,''
  \emph{IEEE Transactions on Information Theory}, vol.~40, no.~5, pp.
  1647--1649, Sep 1994.

\bibitem{Yannis88}
Y.~G. Yatracos, ``A lower bound on the error in non parametric regression type
  problems,'' \emph{Annals of statistics}, vol.~16, no.~3, pp. 1180--1187, Sep
  1988.

\bibitem{Ibragimov81}
I.~A. Ibragimov and R.~Khas'minskii, \emph{Statistical estimation: Asymptotic
  theory}.\hskip 1em plus 0.5em minus 0.4em\relax Springer, New York, 1981.

\bibitem{Wakin06}
M.~Wakin, M.~Duarte, D.~Baron, and R.~Baraniuk, ``Random filters for
  compressive sampling and reconstruction,'' in \emph{Proceedings of the 2006
  IEEE International Conference on Acoustics, Speech, and Signal Processing},
  Toulouse, France, May 2006.

\bibitem{cover}
T.~M. Cover and J.~Thomas, \emph{Elements of Information Theorys}, ser. Wiley,
  New York, 1991.

\bibitem{zamirIT02}
Z.~Reznic, R.~Zamir, and M.~Feder, ``Joint source-channel coding of a gaussian
  mixture source over a gaussian broadcast channel,'' \emph{IEEE Transactions
  on Information Theory}, pp. 776--781, March 2002.

\end{thebibliography}

\end{document}